\documentclass[a4paper,UKenglish,cleveref, autoref, thm-restate]{lipics-v2021}
\nolinenumbers

\usepackage{algorithm}
\newcommand{\red}{\color{black}}

\title{Rewindable Quantum Computation and Its Equivalence to Cloning and Adaptive Postselection}

\author{Ryo Hiromasa\footnote{Hiromasa.Ryo@aj.MitsubishiElectric.co.jp}}{Mitsubishi Electric Corporation, Information Technology R\&D Center, 5-1-1 Ofuna, Kamakura, Kanagawa 247-8501, Japan}{}{}{}
\author{Akihiro Mizutani\footnote{mizutani@eng.u-toyama.ac.jp, This work was done in part while AM worked for Mitsubishi Electric.}}{Faculty of Engineering, University of Toyama, Gofuku 3190, Toyama 930-8555, Japan}{}{}{}
\author{Yuki Takeuchi\footnote{yuki.takeuchi@ntt.com}}
{NTT Communication Science Laboratories, NTT Corporation, 3-1 Morinosato Wakamiya, Atsugi, Kanagawa 243-0198, Japan\\ NTT Research Center for Theoretical Quantum Information, NTT Corporation, 3-1 Morinosato Wakamiya, Atsugi, Kanagawa 243-0198, Japan}{}{}{}
\author{Seiichiro Tani\footnote{seiichiro.tani@ntt.com}}{NTT Communication Science Laboratories, NTT Corporation, 3-1 Morinosato Wakamiya, Atsugi, Kanagawa 243-0198, Japan\\ NTT Research Center for Theoretical Quantum Information, NTT Corporation, 3-1 Morinosato Wakamiya, Atsugi, Kanagawa 243-0198, Japan\\ International Research Frontiers Initiative (IRFI), Tokyo Institute of Technology, Japan}{}{}{}

\authorrunning{R. Hiromasa, A. Mizutani, Y. Takeuchi, and S. Tani}

\Copyright{Ryo Hiromasa, Akihiro Mizutani, Yuki Takeuchi, and Seiichiro Tani}

\ccsdesc[100]{Theory of computation $\rightarrow$ Quantum complexity theory}

\keywords{Quantum computing, Postselection, Lattice problems}

\relatedversion{\red A shortened version of this paper appeared in {\it Proceedings of the 18th Conference on the Theory of Quantum Computation, Communication and Cryptography (TQC 2023)}, pp. 9:1-9:23, 2023~\cite{HMTT23}.
As a difference from the conference paper, this paper includes the detail of rewindable Clifford and instantaneous quantum polynomial time circuits.}

\EventEditors{John Q. Open and Joan R. Access}
\EventNoEds{2}
\EventLongTitle{42nd Conference on Very Important Topics (CVIT 2016)}
\EventShortTitle{CVIT 2016}
\EventAcronym{CVIT}
\EventYear{2016}
\EventDate{December 24--27, 2016}
\EventLocation{Little Whinging, United Kingdom}
\EventLogo{}
\SeriesVolume{42}
\ArticleNo{23}
\hideLIPIcs

\begin{document}
\maketitle

\begin{abstract}
We define rewinding operators that invert quantum measurements.
Then, we define complexity classes ${\sf RwBQP}$, ${\sf CBQP}$, and ${\sf AdPostBQP}$ as sets of decision problems solvable by polynomial-size quantum circuits with a polynomial number of rewinding operators, cloning operators, and adaptive postselections, respectively.
Our main result is that ${\sf BPP}^{\sf PP}\subseteq{\sf RwBQP}={\sf CBQP}={\sf AdPostBQP}\subseteq{\sf PSPACE}$.
As a byproduct of this result, we show that any problem in ${\sf PostBQP}$ can be solved with only postselections of {\red events that occur with probabilities} polynomially close to one.
Under the strongly believed assumption that ${\sf BQP}\nsupseteq{\sf SZK}$, or the shortest independent vectors problem cannot be efficiently solved with quantum computers, we also show that a single rewinding operator is sufficient to achieve tasks that are intractable for quantum computation.
{\red Finally, we show that rewindable Clifford circuits remain classically simulatable, but rewindable instantaneous quantum polynomial time circuits can solve any problem in ${\sf PP}$.}
\end{abstract}

\section{Introduction}
\subsection{\red Background}
It is believed that universal quantum computers outperform their classical counterparts.
There are two approaches to strengthening this belief.
The first is to introduce tasks that seem intractable for classical computers but can be efficiently solved with quantum computers.
For example, no known efficient classical algorithm can solve the integer factorization, but Shor's quantum algorithm~\cite{Sh97} can do it efficiently.
The second approach is to consider what happens if classical computers can efficiently simulate the behaviors of quantum computers.
So far, sampling tasks have often been considered in this approach~\cite{HM17}.
It has been shown that if any probability distribution obtained from some classes of quantum circuits (e.g., instantaneous quantum polynomial time (IQP) circuits~\cite{BJS11}) can be efficiently simulated with classical computers, then ${\sf PH}$ collapses to its second~\cite{FKMNTT18,MTN18} or third level~\cite{BJS11,AA11,MFF14,LLKROR14,BMS16,TTYT16,TT16,GWD17,HKSBSJ17,MSM17,BISBDJBMN18,HBSE18,BFNV19}, or ${\sf BQP}$ is in the second level of ${\sf PH}$~\cite{MTT19}.
Since the collapse of ${\sf PH}$ and the inclusion of ${\sf BQP}$ in ${\sf PH}$ are considered to be unlikely, these results imply quantum advantages.

{\red\subsection{Our Contribution}}
In this paper, we take the second approach.
If efficient classical simulation of quantum measurements is possible, then the measurements become invertible because classical computation {\red can be made} reversible~{\red\footnote{\red More concretely, if quantum measurements can be represented as logic circuits composed of AND and NOT, we can make them invertible by replacing each AND with the Toffoli gate and an ancillary bit.}}.
From the analogy of the rewinding technique used in zero-knowledge (see e.g., \cite{W09,ARU14,U16}), we call such measurements rewindable measurements.
They make quantum computation genuinely reversible and {\it incredibly} powerful{~\red\footnote{\red This result would imply the hardness on the efficient classical simulation of the whole behavior of quantum computers but does not say anything about the separation between ${\sf BQP}$ and ${\sf BPP}$.}}.
More formally, the following rewinding operator $R$ becomes possible. $R$ receives a post-measurement $n$-qubit quantum state $(|z\rangle\langle z|\otimes I^{\otimes n-1})|\psi\rangle$ with $z\in\{0,1\}$ and a {\red succinct} classical description $\mathcal{D}$ of a pre-measurement quantum state $|\psi\rangle$ and outputs the quantum state $|\psi\rangle$:
\begin{eqnarray}
\label{rewind}
R\left(\left(|z\rangle\langle z|\otimes I^{\otimes n-1}\right)|\psi\rangle\otimes|\mathcal{D}\rangle\right)=|\psi\rangle,
\end{eqnarray}
where $I\equiv|0\rangle\langle 0|+|1\rangle\langle 1|$ is the two-dimensional identity operator.
As an important point, $R$ requires the classical description $\mathcal{D}$ as an input.
If it requires only $(|z\rangle\langle z|\otimes I^{\otimes n-1})|\psi\rangle$ as an input, the output state cannot be uniquely determined.
For example, in the case of both $|\psi\rangle=|0\rangle|+\rangle$ and $(|0\rangle|+\rangle+|1\rangle|-\rangle)/\sqrt{2}$, the post-measurement state is $|0\rangle|+\rangle$ for $z=0$, where $|\pm\rangle\equiv(|0\rangle\pm|1\rangle)/\sqrt{2}$.
To circumvent this problem, we require the classical description $\mathcal{D}$ as information about $|\psi\rangle$.
As a concrete example, the classical descriptions of $|0\rangle|+\rangle$ and $(|0\rangle|+\rangle+|1\rangle|-\rangle)/\sqrt{2}$ are $I\otimes H$ and $CZ(H\otimes H)$, respectively.
Here, $H\equiv|+\rangle\langle 0|+|-\rangle\langle1|$ is the Hadamard gate, $CZ\equiv|0\rangle\langle 0|\otimes I+|1\rangle\langle 1|\otimes Z$ is the controlled-$Z$ ($CZ$) gate, and $Z\equiv|0\rangle\langle 0|-|1\rangle\langle 1|$ is the Pauli-$Z$ operator.
These descriptions are proper because $|0\rangle|+\rangle$ and $(|0\rangle|+\rangle+|1\rangle|-\rangle)/\sqrt{2}$ can be prepared by applying $I\otimes H$ and $CZ(H\otimes H)$ on the fixed initial state $|00\rangle$, respectively.
Furthermore, we define rewinding operators for only pure states, {\red and hence} their functionality is arbitrary for mixed states.
Due to this restriction, we can avoid contradictions with an ordinary interpretation of mixed states (see Sec.~\ref{IIIA}) and the no-signaling principle (see {\red Sec.}~\ref{B}).

It is strongly believed that the rewinding of measurements cannot be performed in ordinary quantum mechanics, i.e., the superposition is destroyed by measurements, and it cannot be recovered after measurements.
One may think that if the rewinding were possible against this belief, it could add extra computation power to universal quantum computers.
We show that this expectation is indeed correct.
More formally, we define ${\sf RwBQP}$ (${\sf BQP}$ with rewinding) as a set of decision problems solvable by polynomial-size quantum circuits with a polynomial number of rewinding operators and show ${\sf BQP}\subseteq{\sf BPP}^{\sf PP}\subseteq{\sf RwBQP}$.

The rewinding operator can be considered a probabilistic postselection.
By just repeating measurements and rewinding operations until the target outcome is obtained, we can efficiently simulate the postselection with high probability if the output probability of the target outcome is at least the inverse of some polynomial.
However, the original postselection enables us to deterministically obtain a target outcome even if the probability is exponentially small~\cite{A05}.
In this case, the above simple repeat-until-success approach requires an exponential number of rewinding operations on average.
Surprisingly, we show that it is possible to exponentially mitigate probabilities of nontarget outcomes with a polynomial number of rewinding operators.
By using this mitigation protocol, we can obtain the target outcome with high probability even if the output probability of the target outcome is exponentially small.
In this sense, the rewinding and postselection are equivalent.
More formally, we show that ${\sf RwBQP}$ is equivalent to the class ${\sf AdPostBQP}$ (${\sf BQP}$ with adaptive postselection) of decision problems solvable by polynomial-size quantum circuits with a polynomial number of adaptive postselections.
Here, an adaptive postselection is a projector $|b\rangle\langle b|$ such that the value of $b\in\{0,1\}$ depends on previous measurement outcomes.
From this equivalence, we also obtain ${\sf RwBQP}\subseteq{\sf PSPACE}$.

The rewinding is also related to cloning.
By strengthening our rewinding operator in Eq.~(\ref{rewind}), we define the cloning operator $C$ as follows:
\begin{eqnarray}
\label{clone}
C|\mathcal{D}\rangle=|\psi\rangle.
\end{eqnarray}
Unlike Eq.~(\ref{rewind}), this operator does not require the post-measurement state $(|z\rangle\langle z|\otimes I^{\otimes n-1})|\psi\rangle$.
Since it is easy to duplicate the classical description $\mathcal{D}$, we can efficiently duplicate $|\psi\rangle$, i.e., generate $|\psi\rangle^{\otimes 2}$ by simply applying $C\otimes C$ on $|\mathcal{D}\rangle^{\otimes 2}$.
Although the ordinary cloning operator $\tilde{C}$ is defined such that $\tilde{C}(|\psi\rangle|0^m\rangle)=|\psi\rangle^{\otimes 2}$ for some $m\in\mathbb{N}$~\cite{WZ82}, we define $C$ as an operator whose input is the classical description $\mathcal{D}$ of $|\psi\rangle$ rather than $|\psi\rangle$ itself.
This makes sense because we can always obtain a classical description of $|\psi\rangle$ in our setting\footnote{\red In this paper, we focus on complexity classes based on quantum circuits in uniform families. The extension to those involving prover(s) or advice is mentioned in Sec.~\ref{conclusion} as an open problem.}.
Note that it could be difficult to realize $C$ in quantum polynomial time because $|\psi\rangle$ might be prepared by using measurements.
More precisely, it may be defined as a quantum state prepared when the measurement outcome is $0$, e.g., $|\psi\rangle=U_2(I\otimes|0\rangle\langle 0|\otimes I^{\otimes n-2})U_1|0^n\rangle$ for some unitary operators $U_1$ and $U_2$.
We show that ${\sf RwBQP}$ is also equivalent to the class ${\sf CBQP}$ (${\sf BQP}$ with cloning) of decision problems solvable by polynomial-size quantum circuits with a polynomial number of cloning operators.
That is, the difference between Eqs.~(\ref{rewind}) and (\ref{clone}) does not matter to the computation power.
The following theorem summarizes our main results explained above:
\begin{result}[Theorem~\ref{BRQP=PP}]
\label{theorem1informal}
${\sf BPP}^{\sf PP}\subseteq{\sf RwBQP}={\sf CBQP}={\sf AdPostBQP}\subseteq{\sf PSPACE}$.
\end{result}
The computation power of the cloning has been addressed in \cite{ABFL14} as an open problem.
Result~\ref{theorem1informal} gives lower and upper bounds on our class ${\sf CBQP}$, and it seems to be a reasonable approach to capturing the power of cloning.
Note that more precisely speaking, in Corollary~\ref{nlemma4}, we obtain the lower bound ${\sf BQP}^{\sf PP}_{{\rm classical}}$, which may be slightly tighter than ${\sf BPP}^{\sf PP}$.
Here, the subscript ``classical'' is a symbol used in \cite{A05} and means that only classical queries are allowed.

All the above results assume that rewinding operators can be utilized a polynomial number of times.
Under the strongly believed assumption that the shortest independent vectors problem (SIVP)~\cite{R05} cannot be efficiently solved with universal quantum computers, we show that a single rewinding operator is sufficient to achieve a task that is intractable for universal quantum computation:
\begin{result}[Informal Version of Theorem~\ref{easy}]
\label{theorem2informal}
Assume that there is no polynomial-time quantum algorithm that solves the SIVP.
Then, there exists a problem such that it can be efficiently solved with a constant probability if a single rewinding operator is allowed for quantum computation, but the probability is super-polynomially small if it is not allowed.
\end{result}

We also show a superiority of a single rewinding operator under a different assumption:
\begin{result}[Informal Version of Corollary~\ref{new23}]
\label{result3}
Let ${\sf RwBQP}(1)$ be ${\sf RwBQP}$ with a single rewinding operator.
Then, ${\sf RwBQP}(1)\supset{\sf BQP}$ unless ${\sf BQP}\supseteq{\sf SZK}$.
\end{result}
It is strongly believed that ${\sf BQP}$ does not include ${\sf SZK}$.
At least, we can say that it is hard to show ${\sf BQP}\supseteq{\sf SZK}$ because there exists an oracle $A$ such that ${\sf BQP}^A\nsupseteq{\sf SZK}^A$~\cite{A02}.
For example, by assuming that the decision version of SIVP, gapSIVP, is hard for universal quantum computation, Result~\ref{result3} implies that a single rewinding operator is sufficient to achieve a task that is intractable for universal quantum computation.
This is because the gapSIVP (with an appropriate parameter) is in ${\sf SZK}$~\cite{PV08}.
As a difference from Result~\ref{theorem2informal}, Result~\ref{result3} shows the superiority of a single rewinding operator for decision problems.

As simple observations, we also consider the effect of rewinding operators for restricted classes of quantum circuits.
It has been shown that polynomial-size Clifford circuits are classically simulatable~\cite{G99}.
In {\red Sec.}~\ref{A}, we show that such circuits with rewinding operators are still classically simulatable.
It is also known that IQP circuits are neither universal nor classically simulatable under plausible complexity-theoretic assumptions~\cite{BMS16}.
In {\red Sec.}~\ref{B}, we show that IQP circuits with rewinding operators can efficiently solve any problem in ${\sf PP}$.

Our mitigation protocol used to show ${\sf AdPostBQP}\subseteq{\sf RwBQP}$ also has an application for ${\sf PostBQP}$~\cite{A05}, which is a class of decision problems solvable by polynomial-size quantum circuits with non-adaptive postselections.
By slightly modifying our mitigation protocol and replacing rewinding operators with postselections, we obtain the following corollary:
\begin{result}[Corollary~\ref{main1}]
\label{coroinfo}
For any polynomial function {\red$P(|x|)$} in the size $|x|$ of an instance $x$, ${\sf PP}={\sf PostBQP}$ holds even if only non-adaptive postselections of outputs whose probabilities are {\red $1-O(1/P(|x|))$} are allowed.
 \end{result}
The equality ${\sf PP}={\sf PostBQP}$ was originally shown in \cite{A05} by using postselections of outputs whose probabilities may be exponentially small.
Result~\ref{coroinfo} shows that such postselections can be replaced with those of outputs whose probabilities are polynomially close to one.
This result is {\it optimal} in the sense that polynomially many repetitions of non-adaptive postselections of outputs whose probabilities are $1-1/f(|x|)$ with a super-polynomial function $f(|x|)$ can be simulated in quantum polynomial time.
It is worth mentioning that when the probabilities are at least some constant, the above replacement is obvious in ${\sf PostBPP}$ (or ${\sf BPP}_{\sf path}$).
This is because any behavior of a probabilistic Turing machine can be represented as a binary tree such that each path is chosen with probability $1/2$.
However, in its quantum analogue ${\sf PostBQP}$, it was open as to whether such replacement is possible even if the probabilities are at least some constant.

\medskip
\noindent
{\bf Related Work.}
As a hypothetical ability that enables us to perform the cloning, a closed timelike curve (CTC) has already been studied~\cite{BWW13}.
It also rewinds the time in a sense.
Among several formulations of the CTC~\cite{D91,LMPGSPRDSSS11}, the Deutschian CTC~\cite{D91} is the major one.
Let $\mathcal{R}_{\rm CTC}$ and $\mathcal{Q}$ be quantum registers that are input to and not input to the CTC, respectively, and ${\rm Tr}_{\mathcal{R}}[\cdot]$ be the partial trace over the system $\mathcal{R}$.
Suppose that the register $\mathcal{Q}$ is initialized to $|0^m\rangle\langle 0^m|$ with $m\in\mathbb{N}$.
The Deutschian CTC prepares a quantum state\footnote{A fixed-point theorem guarantees the existence of such $\rho$ for any $U$~\cite{D91}.} $\rho$ on $\mathcal{R}_{\rm CTC}$ such that for any unitary operator $U$ acting on $\rho\otimes|0^m\rangle\langle 0^m|$,
\begin{eqnarray}
\label{CTC}
{\rm Tr}_{\mathcal{Q}}\left[U\left(\rho\otimes|0^m\rangle\langle 0^m|\right)U^\dag\right]=\rho.
\end{eqnarray}
Equation~(\ref{CTC}) means that the future state $\rho$ generated after applying $U$ is prepared before applying $U$ under the constraint of the causal consistency.
Aaronson and Watrous have shown that classes ${\sf P}_{\sf CTC}$ and ${\sf BQP}_{\sf CTC}$ of decision problems solvable by polynomial-size logic and quantum circuits with the Deutschian CTC are equivalent and that they are also equivalent to ${\sf PSPACE}$~\cite{AW08}.
From Result~\ref{theorem1informal}, the Deutschian CTC should be at least as powerful as the rewinding.

As another hypothetical ability, non-collapsing measurements, which allow us to obtain measurement outcomes without perturbing quantum states, have been considered in \cite{ABFL14,ABFL16}.
For the class ${\sf PDQP}$ of decision problems solvable by polynomial-size quantum circuits with non-collapsing measurements, Aaronson {\it et al.} have shown ${\sf SZK}\subseteq{\sf PDQP}\subseteq{\sf BPP}^{\sf PP}$~\cite{ABFL14}.
Therefore, Result~\ref{theorem1informal} implies that the rewinding should be at least as powerful as the non-collapsing measurements.

\subsection{Overview of Techniques}
\label{IB}
To obtain Result~\ref{theorem1informal}, we show (i) ${\sf RwBQP}\subseteq{\sf CBQP}$; (ii) ${\sf CBQP}\subseteq{\sf AdPostBQP}$; (iii) ${\sf AdPostBQP}\subseteq{\sf RwBQP}$; (iv) ${\sf BQP}^{\sf PP}_{{\rm classical}}\subseteq{\sf RwBQP}$, which immediately means ${\sf BPP}^{\sf PP}\subseteq{\sf RwBQP}$ because ${\sf BPP}^{\sf PP}\subseteq{\sf BQP}^{\sf PP}_{{\rm classical}}$; and (v) ${\sf AdPostBQP}\subseteq{\sf PSPACE}$.
The first inclusion (i) is obvious from the definitions of the rewinding operator $R$ and cloning operator $C$ (see Eqs.~(\ref{rewind}) and (\ref{clone})).
The fifth inclusion (v) can also be easily shown by using the Feynman path integral that is used to show ${\sf BQP}\subseteq{\sf PSPACE}$~\cite{BV97}.
In ${\sf BQP}$, measurements are only performed at the end of quantum circuits.
On the other hand, in ${\sf AdPostBQP}$, intermediate ordinary and postselection measurements are also allowed.
However, this difference does not matter in showing the inclusion in ${\sf PSPACE}$.

The second inclusion (ii) is a natural consequence from the simple observation that postselections can simulate the cloning operator $C$.
On the other hand, the third inclusion (iii) is nontrivial because we have to efficiently simulate postselection by using only a polynomial number of rewinding operators.
To this end, we give an efficient protocol to exponentially mitigate the amplitude of a nontarget state by using a polynomial number of rewinding operators.
Let $|\psi\rangle=\sqrt{2^{-p(n)}}|\psi_t\rangle+\sqrt{1-2^{-p(n)}}|\psi_t^\perp\rangle$, where $p(n)$ is some polynomial in the size $n$ of a given ${\sf AdPostBQP}$ problem, $|\psi_t\rangle$ is a target state that we would like to postselect, and $\langle\psi_t|\psi_t^\perp\rangle=0$.
By using our mitigation protocol, from $|\psi\rangle$, we can obtain $\sqrt{2^{-p(n)}}|\psi_t\rangle+\sqrt{2^{-p(n)}(1-2^{-p(n)})}|\psi_t^\perp\rangle$
up to a normalization factor.
Since $2^{-p(n)}$ is larger than $2^{-p(n)}(1-2^{-p(n)})$, we now obtain $|\psi_t\rangle$ with probability of at least $1/2$.
By repeating these procedures, we can simulate the postselection of $|\psi_t\rangle$ with high probability. 

Our mitigation protocol is also useful in showing the fourth inclusion (iv).
First, from ${\sf PP}={\sf PostBQP}$~\cite{A05}, we obtain ${\sf PP}\subseteq{\sf RwBQP}$ by using our mitigation protocol.
Then, we show that the completeness-soundness gap in ${\sf RwBQP}$ can be amplified to a value exponentially close to $1$, and ${\sf RwBQP}$ is closed under composition if only classical queries are allowed.
By combining these results, we obtain ${\sf BQP}^{\sf PP}_{{\rm classical}}\subseteq{\sf RwBQP}^{\sf RwBQP}_{{\rm classical}}={\sf RwBQP}$.
Note that ${\sf BQP}^{\sf PP}_{{\rm classical}}\subseteq{\sf RwBQP}^{\sf PP}_{{\rm classical}}$ is obvious from the definition of ${\sf RwBQP}$ (see Def.~\ref{BRQP}).

We show Result~\ref{theorem2informal} as follows.
Cojocaru {\it et al.} have shown that under the hardness of SIVP, there exists a family $\mathcal{F}\equiv\{f_K\}_{K\in\mathcal{K}}$ of functions that is collision resistant against quantum computers, i.e., no polynomial-time quantum algorithm can output a collision pair $(x,x')$ such that $x\neq x'$ and $f_K(x)=f_K(x')$~\cite{CCKW18}.
Here, $\mathcal{K}$ is a finite set of parameters uniquely specifying each function (see Sec.~\ref{IIC} for details).
We show that a collision pair can be output with a constant probability if only one rewinding operator is given.
From the construction of $\mathcal{F}$, the last bits of collision pairs are different, i.e., there exist $x_0$ and $x_1$ such that $x=(x_0,0)$ and $x'=(x_1,1)$.
Using the idea in \cite{CCKW19}, we can efficiently prepare
\begin{eqnarray}
\label{preimage}
\frac{|x_0\rangle|0\rangle+|x_1\rangle|1\rangle}{\sqrt{2}}
\end{eqnarray}
for some output value $y=f_K(x)=f_K(x')$.
Note that since the preparation of Eq.~(\ref{preimage}) includes a measurement, if we perform it again, we will obtain a quantum state in Eq.~(\ref{preimage}) for a different output value $y'$, and hence it is difficult to simultaneously obtain $x$ and $x'$ for the same $y$.
When we can use a rewinding operator, the situation changes.
By measuring the state in Eq.~(\ref{preimage}), we can obtain $x_0$ or $x_1$.
For simplicity, suppose that we obtain $x_0$.
Then, by performing the rewinding operator $R$ on $|x_0\rangle|0\rangle$ and a classical description of Eq.~(\ref{preimage}), we can prepare the quantum state in Eq.~(\ref{preimage}) for the {\it same} $y$.
From this state, we can obtain $x_1$ with probability $1/2$.
As an important point, since the last bits of $x$ and $x'$ differ, a single rewinding operator (i.e., the rewinding of a single qubit) is sufficient to find a collision pair with a constant probability.

Finally, we show Result~\ref{result3}.
To this end, we show that a ${\sf SZK}$-complete problem is in ${\sf RwBQP}(1)$ by using a technique inspired by \cite{ABFL16}.

\section{Preliminaries}
\label{II}
In this section, we review some preliminaries that are necessary to understand our results.
In Sec.~\ref{IIA}, we introduce a complexity class ${\sf PostBQP}$ and explain the postselection.
In Sec.~\ref{IIC}, we introduce the SIVP and a collision-resistant and $\delta-2$ regular family of functions.

\subsection{Quantum {\red C}omplexity {\red C}lass}
\label{IIA}
In this subsection, we review ${\sf PostBQP}$ and explain the postselection.
Then, we clarify a difference between ${\sf PostBQP}$ and our class ${\sf AdPostBQP}$ (see Def.~\ref{AdPostBQP}).
Note that we assume that readers are familiar with classical complexity classes~\cite{AB09}.
${\sf PostBQP}$ is defined as follows:
\begin{definition}[${\sf PostBQP}$~\cite{A05}]
\label{postbqpa}
A promise problem $L=(L_{\rm yes},L_{\rm no})\subseteq\{0,1\}^\ast$ is in ${\sf PostBQP}$ if and only if there exist polynomials $n$ and $q$ and a uniform family $\{U_x\}_x$ of polynomial-size quantum circuits, such that
\begin{itemize}
\item ${\rm Pr}[p=1]\ge1/2^q$ 
\item when $x\in L_{\rm yes}$, ${\rm Pr}[o=1\ |\ p=1]\ge2/3$
\item when $x\in L_{\rm no}$, ${\rm Pr}[o=1\ |\ p=1]\le1/3$,
\end{itemize}
where $o$ and $p$ are called {\red single-qubit} output and postselection registers, respectively.
Here, for any $z_1\in\{0,1\}$ and $z_2\in\{0,1\}$,
\begin{eqnarray}
{\rm Pr}[p=z_2]&\equiv&\langle 0^n|U_x^\dag\left(I\otimes |z_2\rangle\langle z_2|\otimes I^{\otimes n-2}\right)U_x|0^n\rangle\\
{\rm Pr}[o=z_1\ |\ p=z_2]&\equiv&\cfrac{\langle 0^n|U_x^\dag\left(|z_1z_2\rangle\langle z_1z_2|\otimes I^{\otimes n-2}\right)U_x|0^n\rangle}{{\rm Pr}[p=z_2]}.
\end{eqnarray}
In this definition, ``polynomial'' means the one in the length $|x|$ of the instance $x$.
\end{definition}
From Def.~\ref{postbqpa}, we notice that the postselection is to apply a projector.
In ${\sf PostBQP}$, it is allowed to apply the projector $|1\rangle\langle 1|$ to the qubit in the postselection register at the end of a quantum circuit.
Therefore, ${\sf PostBQP}$ is a set of promise problems solvable by polynomial-size quantum circuits (in uniform families) with a single non-adaptive postselection\footnote{Note that a polynomial number of postselections are allowed if they can be unified as a single non-adaptive postselection.}.
On the other hand, in ${\sf AdPostBQP}$, we allow the application of a polynomial number of intermediate measurements and projectors.
This means that the value $b\in\{0,1\}$ of a projector $|b\rangle\langle b|$ can depend on previous measurement outcomes, while it is determined before executing a quantum circuit in ${\sf PostBQP}$.

\subsection{\red Shortest Independent Vectors Problem (SIVP)}
\label{IIC}
The SIVP with approximation factor $\gamma$ (${\rm SIVP}_\gamma$) is defined as follows:
\begin{definition}[${\rm SIVP}_\gamma$]
Let $n$ be any natural number and $\gamma\ (\ge1)$ be any real number.
Given an $n$ bases of a lattice $L$, output a set of $n$ linearly independent lattice vectors of length at most $\gamma\cdot\lambda_n(L)$.
Here, $\gamma$ can depend on $n$, and $\lambda_n(L)$ is the $n$th successive minimum of $L$ (i.e., the smallest $r$ such that $L$ has $n$ linearly independent vectors of norm at most $r$).
\end{definition}
Since there is no known polynomial-time quantum algorithm to solve ${\rm SIVP}_\gamma$ for polynomial approximation factor, it is used as a basis of the security of lattice-based cryptography~\cite{R05}.

The hardness of the SIVP is also used to construct families of collision-resistant functions against universal quantum computers.
From \cite{CCKW18}, we can immediately obtain the following theorem:
\begin{theorem}[adapted from \cite{CCKW18}]
\label{crf}
Let $n$ be any natural number, $q=2^{5\lceil\log_2{n}\rceil+21}$, $m=23n+5n\lceil\log_2{n}\rceil$, $\mu=2mn\sqrt{23+5\log_2{n}}$, and $\mu'=\mu/m$, where $\lceil\cdot\rceil$ is the ceiling function.
Let $K\equiv(A,As_0+e_0)\in\mathcal{K}$ with $\mathcal{K}$ being the multiset $\{(A,As_0+e_0)\}_{A\in\mathbb{Z}_q^{n\times m},s_0\in\mathbb{Z}_q^n,e_0\in{\chi'}^m}$, where $\mathbb{Z}_q^{n\times m}$ be the set of $n\times m$ matrices each of whose entry is chosen from $\mathbb{Z}_q\equiv\{0,1,\ldots,q-1\}$, and $\chi'$ is the set of integers bounded in absolute value by $\mu'$.
Assume that there is no polynomial-time quantum algorithm that solves ${\rm SIVP}_{p(n)}$ for some polynomial $p(n)$ in $n$.
Then, the family $\mathcal{F}\equiv\{f_K:\mathbb{Z}_q^n\times\chi^m\times\{0,1\}\rightarrow\mathbb{Z}_q^m\}_{K\in\mathcal{K}}$ of functions
\begin{eqnarray}
\label{func}
f_K(s,e,c)\equiv As+e+c\cdot(As_0+e_0)\ ({\rm mod}\ q),
\end{eqnarray}
where $\chi$ is the set of integers bounded in absolute value by $\mu$, is collision resistant\footnote{Let $\mathcal{F}\equiv\{f_K:\mathcal{D}\rightarrow\mathcal{R}\}_{K\in\mathcal{K}}$ be a function family. We say that $\mathcal{F}$ is collision resistant if for any polynomial-time quantum algorithm $A$, which receives $K$ and outputs two bit strings $x,x'\in\mathcal{D}$, the probability ${\rm Pr}_K[A(K)=(x,x')\ {\rm such}\ {\rm that}\ x\neq x'\ {\rm and}\ f_K(x)=f_K(x')]$ is super-polynomially small. Note that $K$ is chosen from $\mathcal{K}$ uniformly at random, and the probability is also taken over the randomness in $A$.} and $\delta$-$2$ regular\footnote{Let $\mathcal{F}\equiv\{f_K:\mathcal{D}\rightarrow\mathcal{R}\}_{K\in\mathcal{K}}$ be a function family. For a fixed $K$, we say that $y\in\mathcal{R}$ has two preimages if there exist exactly two different inputs $x,x'\in\mathcal{D}$ such that $f(x)=f(x')=y$. Let $\mathcal{Y}_K^{(2)}$ be the set of $y$ having two preimages for $K$. The function family $\mathcal{F}$ is said to be $\delta$-$2$ regular when ${\rm Pr}_{K,x}[f_K(x)\in\mathcal{Y}_K^{(2)}]\ge\delta$, where $K$ and $x$ are chosen from $\mathcal{K}$ and $\mathcal{D}$, respectively, uniformly at random.} for a constant $\delta$.
\end{theorem}
From Eq.~(\ref{func}), the function  $f_K$ has a collision pair\footnote{Since $q$ is larger than $\mu$, the second element $e+e_0$ of the second input may not be in the set $\chi^m$. Therefore, the probability of $f_K$ having a collision pair is not $1$.} $(s,e,1)$ and $(s+s_0,e+e_0,0)$, and Theorem~\ref{crf} shows that it is difficult to find them simultaneously.
This function family will be used to show that a single rewinding operator is sufficient to achieve a task that seems difficult for universal quantum computers.

Note that in \cite{CCKW18}, the matrix $A$ is constructed so that it has a trapdoor to efficiently invert $As+e$, and its distribution is statistically close to uniform over $\mathbb{Z}_q^{n\times m}$.
In Theorem~\ref{crf}, we consider a simplified variant of the function family of \cite{CCKW18} in which the matrix $A$ is chosen uniformly at random.

\section{Computational Complexity of Rewinding}
\label{IIIA}
In this section, we show Results~\ref{theorem1informal} and \ref{coroinfo}.

{\red\subsection{\red Our Complexity Classes}}
To this end, first, we define the rewinding operator $R$ and cloning operator $C$ as follows:
\begin{definition}[Rewinding and Cloning Operators]
\label{Rewinding}
Let $n$ be any natural number, $Q$ be any $n$-qubit linear operator composed of unitary operators and the $Z$-basis projective operators $\{|0\rangle\langle 0|,|1\rangle\langle 1|\}$, $\mathcal{D}$ be a classical description of the linear operator $Q$, and $I$ be the single-qubit identity operator.
The rewinding and cloning operators $R$ and $C$ are maps from a quantum state to a quantum state such that for any $s\in\{0,1\}$, when $\left(|s\rangle\langle s|\otimes I^{\otimes n-1}\right)Q|0^n\rangle\neq 0$,
\begin{eqnarray}
\label{rewinddef}
R\left(\frac{\left(|s\rangle\langle s|\otimes I^{\otimes n-1}\right)Q|0^n\rangle}{\sqrt{\langle 0^n|Q^\dag\left(|s\rangle\langle s|\otimes I^{\otimes n-1}\right)Q|0^n\rangle}}\otimes|\mathcal{D}\rangle\right)&=&\frac{Q|0^n\rangle}{\sqrt{\langle 0^n|Q^\dag Q|0^n\rangle}}\\
C|\mathcal{D}\rangle&=&\frac{Q|0^n\rangle}{\sqrt{\langle 0^n|Q^\dag Q|0^n\rangle}}.
\end{eqnarray}
For other input states, the functionality of $R$ and $C$ is undefined, that is outputs are arbitrary $n$-qubit states.
Particularly when it depends on a value of classical bits whether $R$ and $C$ are applied, we call them classically controlled rewinding and cloning operators, respectively.
\end{definition}
Note that since the linear operator $Q$ may include projective operators (e.g., $Q=U_2(I\otimes|0\rangle\langle 0|\otimes I^{\otimes n-2})U_1$ for some $n$-qubit unitary operators $U_1$ and $U_2$), in general, $Q^\dag Q\neq I^{\otimes n}$.
{\red From Def.~\ref{Rewinding}, it is easily observed that the cloning operator $C$ can efficiently simulate the rewinding operator $R$.}

An example of classically controlled rewinding operators is an operator such that if a measurement outcome is $0$, the identity operator is applied to the post-measurement state, but if the outcome is $1$, the rewinding operator $R$ is applied to it.
Classically controlled rewinding and cloning operators play an important role in giving our main {\red results}.
It is worth mentioning that $|\mathcal{D}\rangle$ is consumed by applying a single rewinding or cloning operator.
However, we can always retain $|\mathcal{D}\rangle$ by duplicating it before applying the rewinding or cloning operator, and hence subsequent rewinding or cloning operators can also be applied.
{\red These remarks are explained more explicitly in Fig.~\ref{circuit}.}

Simply speaking, the rewinding operator $R$ rewinds the state projected onto $|s\rangle$ to the state before the measurement.
As an important point, Def.~\ref{Rewinding} implies that the rewinding operator $R$ only works for pure states.
The following contradiction for an ordinary interpretation of mixed states occurs without the restriction to pure states.
Suppose that we measure a maximally mixed state $I/2$ in the computational basis\footnote{We sometimes call the $Z$ basis the computational basis.}, and then obtain the measurement outcome $0$.
In this case, it is natural that even if we rewind this measurement and perform the same measurement again, the outcome is always $0$.
However, if we define the rewinding operator $R$ so that it also works for mixed states, then we can obtain $I/2$ from $|0\rangle$ with the rewinding operator, and the measurement on it may output $1$.
In other words, if the rewinding operator works for a mixed state $\rho$, we can measure $\rho$ again and again, and thus we obtain its information as much as we want without changing $\rho$.
This situation contradicts with the natural interpretation that mixed states arise due to the lack of knowledge about them.
Furthermore, the restriction to pure states would be useful in circumventing the contradiction with the no-signaling principle as explained in {\red Sec.}~\ref{B}{\red.}

By using the rewinding and cloning operators and postselections, we define three complexity classes---${\sf RwBQP}$ (${\sf BQP}$ with rewinding), ${\sf CBQP}$ (${\sf BQP}$ with cloning), and ${\sf AdPostBQP}$ (${\sf BQP}$ with adaptive postselection)---as follows:
\begin{definition}[${\sf RwBQP}$ and ${\sf CBQP}$]
\label{BRQP}
Let $n$ and $k$ be any natural number, $\ell$ be a polynomial in $n$, and $0\le s<c\le1$.
A promise problem $L=(L_{\rm yes},L_{\rm no})\subseteq\{0,1\}^\ast$ is in ${\sf RwBQP}(c,s)(k)$ if and only if there exists a polynomial-time deterministic Turing machine that receives $1^n$ as an input and generates a $\ell$-bit description $\tilde{\mathcal{D}}$ of an operator $Q_n$ such that it consists of a polynomial number of elementary gates in a universal gate set, single-qubit measurements in the computational basis, and $k$ (classically controlled) rewinding operators $R$ defined in Def.~\ref{Rewinding} and satisfies, for the instance $x\in\{0,1\}^n$ and a polynomial $m$, that
\begin{itemize}
\item if $x\in L_{\rm yes}$, $\sum_{z\in A}\left|\left|\left(|1\rangle\langle 1|\otimes I^{\otimes n+m+\ell-1}\right)Q_n^{(z)}\left(|x\rangle|0^m\rangle|\tilde{\mathcal{D}}\rangle\right)\right|\right|^2\ge c$
\item if $x\in L_{\rm no}$, $\sum_{z\in A}\left|\left|\left(|1\rangle\langle 1|\otimes I^{\otimes n+m+\ell-1}\right)Q_n^{(z)}\left(|x\rangle|0^m\rangle|\tilde{\mathcal{D}}\rangle\right)\right|\right|^2\le s$,
\end{itemize}
where $A$ is the set of possible {\red outcomes of intermediate measurements}, $Q_n^{(z)}$ is the same as $Q_n$ except for that the $i$th measurement is replaced with $|z_i\rangle\langle z_i|$ for all $i$, and $z_i$ is the $i$th bit of $z$.
Here, $|||v\rangle||\equiv\sqrt{\langle v|v\rangle}$ for any vector $|v\rangle$, and ``polynomial'' is the abbreviation of ``polynomial in $n$.''
Particularly, for the set ${\rm poly}(n)$ of all polynomial functions, we denote $\bigcup_{k\in{\rm poly}(n)}{\sf RwBQP}(c,s)(k)$ and $\bigcup_{k\in{\rm poly}(n)}{\sf RwBQP}(2/3,1/3)(k)$ as ${\sf RwBQP}(c,s)$ and ${\sf RwBQP}$, respectively.

By replacing $R$ with the cloning operator $C$ defined in Def.~\ref{Rewinding}, ${\sf CBQP}(c,s)(k)$, ${\sf CBQP}(c,s)$, and ${\sf CBQP}$ are defined in a similar way.
\end{definition}
To perform a rewinding operator $R$ to recover an intermediate state $|\psi\rangle$, a classical description $\mathcal{D}$ of $|\psi\rangle$ is neccessary.
It can always be generated from $\tilde{\mathcal{D}}$ and measurement outcomes obtained before preparing $|\psi\rangle$.
As in the case of ${\sf BQP}$, computations performed to solve ${\sf RwBQP}$ problems can be written as uniform families of quantum circuits.

{\red Due to the addition of rewinding operators, it may be difficult to imagine quantum circuits used in ${\sf RwBQP}$.
To clarify them, as an example, we give a concrete circuit diagram for the following ${\sf RwBQP}$ computation.
Suppose that we would like to prepare a two qubit state $(|0\rangle\langle 0|\otimes I)U|00\rangle$ (up to normalization) for a two-qubit unitary operator $U$.
To this end, we use at most two classically controlled rewinding operators.
More precisely, the rewinding operator $R$ is applied if and only if the measurement outcome is $1$.
This computation can be depicted as a fixed quantum circuit in Fig.~\ref{circuit}.}

\begin{figure}[t]
\begin{center}
\includegraphics[width=12cm, clip]{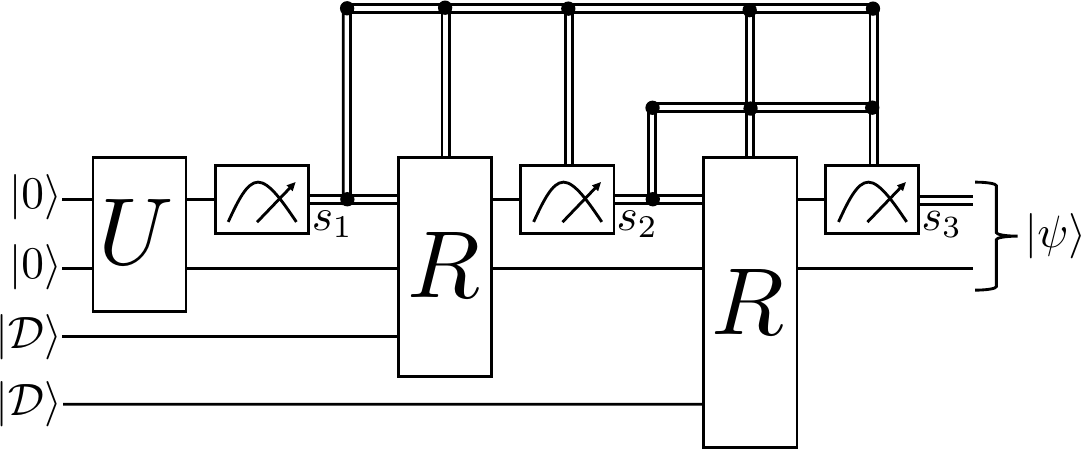}
\end{center}
\caption{Concrete example of ${\sf RwBQP}$ computation.
$\mathcal{D}$, $U$, $R$, and $|\psi\rangle$ are a classical description of $U|00\rangle$, a two-qubit unitary operator, the rewinding operator, and the output state, respectively.
More precisely, when $U=\prod_iu_i$ for elementary quantum gates $u_i$ in a universal gate set, $\mathcal{D}$ is a bit string representing $\prod_iu_i$.
Note that since $|\psi\rangle$ is prepared by using only the unitary operator $U$, its classical description $\mathcal{D}$ does not include projectors and can be generated from only $\tilde{\mathcal{D}}$.
Meter symbols represent the Pauli-$Z$ measurements, and $s_i\in\{0,1\}$ is the $i$th measurement outcome for $1\le i\le 3$.
We represent $|s_i\rangle$ as a classical bit $s_i$ to emphasize that it can be copied.
When the first measurement outcome $s_1$ is $1$, the first rewinding operator $R$ is applied. On the other hand, when $s_1=0$, we do not apply $R$, because the target state is obtained. Since the second and third measurements and the second rewinding operator are applied only when $s_1=1$, they are also conditioned on $s_1$. In a similar way, since it is not necessary to apply the second rewinding operator if $s_2=0$, the second rewinding operator and the third measurement are also conditioned on $s_2$. Finally, $|\psi\rangle$ becomes the target state when $s_1=0$, $(s_1,s_2)=(1,0)$, or $(s_1,s_2,s_3)=(1,1,0)$.}
\label{circuit}
\end{figure}

From Def.~\ref{BRQP}, we immediately obtain the following lemma:
\begin{lemma}
\label{nlemma1}
${\sf RwBQP}\subseteq{\sf CBQP}$.
\end{lemma}
\begin{proof}
The only difference between ${\sf RwBQP}$ and ${\sf CBQP}$ is whether the rewinding or cloning operator is allowed.
Since the cloning operator $C$ can exactly simulate the rewinding operator $R$, this lemma holds.
\end{proof}

\begin{definition}[${\sf AdPostBQP}$]
\label{AdPostBQP}
Let $n$ be any natural number, $\ell$ be a polynomial in $n$, and $0\le s<c\le1$.
A promise problem $L=(L_{\rm yes},L_{\rm no})\subseteq\{0,1\}^\ast$ is in ${\sf AdPostBQP}(c,s)$ if and only if there exists a polynomial-time deterministic Turing machine that receives $1^n$ as an input and generates a $\ell$-bit description $\tilde{\mathcal{D}}$ of an operator $Q_n$ such that it consists of a polynomial number of elementary gates in a universal gate set, single-qubit measurements in the computational basis, and single-qubit projectors $|1\rangle\langle 1|$ and satisfies, for the instance $x\in\{0,1\}^n$ and a polynomial $m$, that
\begin{itemize}
\item if $x\in L_{\rm yes}$, $\sum_{z\in A}q_z\left|\left|\left(|1\rangle\langle 1|\otimes I^{\otimes n+m+\ell-1}\right)\mathcal{N}[Q_n^{(z)}\left(|x\rangle|0^m\rangle|\tilde{\mathcal{D}}\rangle\right)]\right|\right|^2\ge c$
\item if $x\in L_{\rm no}$, $\sum_{z\in A}q_z\left|\left|\left(|1\rangle\langle 1|\otimes I^{\otimes n+m+\ell-1}\right)\mathcal{N}[Q_n^{(z)}\left(|x\rangle|0^m\rangle|\tilde{\mathcal{D}}\rangle\right)]\right|\right|^2\le s$,
\end{itemize}
where $A$ is the set of {\red possible outcomes of intermediate measurements}, $Q_n^{(z)}$ is the same as $Q_n$ except for that the $i$th measurement is replaced with $|z_i\rangle\langle z_i|$ for all $i$, $z_i$ is the $i$th bit of $z$, $q_z$ is the probability of obtaining $z$, and $\mathcal{N}[\cdot]$ denotes the normalization of the vector in the square brackets.
Here, ``polynomial'' is the abbreviation of ``polynomial in $n$.''
Note that for $1\le i\le n$, a projector $|1\rangle\langle 1|_i$ on the $i$th qubit can be applied only when a quantum state $|\psi\rangle$ to be applied satisfies
\begin{eqnarray}
\label{condadpost}
\left|\left|\left(|1\rangle\langle 1|_i\right)|\psi\rangle\right|\right|^2\ge 2^{-p(n)}
\end{eqnarray}
for a polynomial $p(n)$ in $n$.
Particularly, we denote ${\sf AdPostBQP}(2/3,1/3)$ as ${\sf AdPostBQP}$.
\end{definition}
{\red Equation}~(\ref{condadpost}) can be automatically satisfied by using standard gate sets whose elementary gates involve only square roots of rational numbers.
From Defs.~\ref{BRQP} and \ref{AdPostBQP}, we notice that the main difference between ${\sf RwBQP}$, ${\sf CBQP}$, and ${\sf AdPostBQP}$ is whether the rewinding or cloning operators or projectors are allowed.
{\red We also mention a difference between ${\sf AdPostBQP}$ and ${\sf PostBQP}$ as the following remark:}
\begin{remark}
$Q_n$ can include adaptive postselections because depending on previous measurement outcomes, we can decide whether $X$ is applied before and after applying $|1\rangle\langle1|$.
Here, $X\equiv|1\rangle\langle 0|+|0\rangle\langle 1|$ is the Pauli-$X$ operator.
It is worth mentioning that it is unknown whether the adaptive postselection can be efficiently done in ${\sf PostBQP}$ as discussed in \cite{A05}.
Indeed, if it is possible, ${\sf SZK}\subseteq{\sf PP}$ should be immediately obtained from the argument in \cite{ABFL14,ABFL16}, while it is a long-standing problem.
The difference between ${\sf AdPostBQP}$ and ${\sf PostBQP}$ should arise from intermediate measurements allowed in ${\sf AdPostBQP}$.
\end{remark}

The following three corollaries {\red hold}:
\begin{corollary}
\label{corollary}
${\sf RwBQP}$, ${\sf CBQP}$, and ${\sf AdPostBQP}$ are closed under complement.
\end{corollary}

\begin{corollary}
\label{corollary2}
${\sf RwBQP}={\sf RwBQP}(1-2^{-p(n)},2^{-p(n)})$, ${\sf CBQP}={\sf CBQP}(1-2^{-p(n)},2^{-p(n)})$, and ${\sf AdPostBQP}={\sf AdPostBQP}(1-2^{-p(n)},2^{-p(n)})$ for any polynomial function $p(n)$ in the size $n$ of a given instance $x$.
\end{corollary}

\begin{corollary}
\label{corollary3}
When only classical queries are allowed, ${\sf RwBQP}$, ${\sf CBQP}$, and ${\sf AdPostBQP}$ are closed under composition.
{\red In other words, ${\sf RwBQP}^{\sf RwBQP}_{{\rm classical}}={\sf RwBQP}$, ${\sf CBQP}^{\sf CBQP}_{{\rm classical}}={\sf CBQP}$, and ${\sf AdPostBQP}^{\sf AdPostBQP}_{{\rm classical}}={\sf AdPostBQP}$ hold.}
\end{corollary}
Since they are obvious from Defs.~\ref{BRQP} and \ref{AdPostBQP} and can be shown by using standard techniques, proofs are given in Appendix~\ref{E}.

{\red\subsection{\red Relations with Classical Complexity Classes}}
From Def.~\ref{AdPostBQP}, we immediately obtain the following lemma:
\begin{lemma}
\label{lemmapspace}
${\sf AdPostBQP}\subseteq{\sf PSPACE}$.
\end{lemma}
\begin{proof}
The proof is essentially the same as that of ${\sf BQP}\subseteq{\sf PSPACE}$~\cite{BV97}.
The details are given in Appendix~\ref{C}.
\end{proof}

In the rest of this {\red subsection}, we consider a relation between the rewinding, cloning, and postselection (i.e., ${\sf RwBQP}$, ${\sf CBQP}$, and ${\sf AdPostBQP}$), and also obtain lower and upper bounds on them.
More formally, we show the following theorem:
\begin{theorem}
\label{BRQP=PP}
${\sf BPP}^{\sf PP}\subseteq{\sf RwBQP}={\sf CBQP}={\sf AdPostBQP}\subseteq{\sf PSPACE}$.
\end{theorem}
\begin{proof}
This theorem can be obtained by showing (i) ${\sf RwBQP}\subseteq{\sf CBQP}$; (ii) ${\sf CBQP}\subseteq{\sf AdPostBQP}$; (iii) ${\sf AdPostBQP}\subseteq{\sf RwBQP}$; (iv) ${\sf BQP}^{\sf PP}_{{\rm classical}}\subseteq{\sf RwBQP}$, which immediately means ${\sf BPP}^{\sf PP}\subseteq{\sf RwBQP}$ because ${\sf BPP}^{\sf PP}\subseteq{\sf BQP}^{\sf PP}_{{\rm classical}}$; and (v) ${\sf AdPostBQP}\subseteq{\sf PSPACE}$.
The inclusions (i) and (v) are already shown in Lemmas~\ref{nlemma1} and \ref{lemmapspace}, respectively.
The remaining inclusions (ii), (iii), and (iv) will be shown in Lemma~\ref{lemma1} and Corollary~\ref{nlemma4}.
\end{proof}

To simplify our argument in proofs of Lemma~\ref{lemma1} and Theorem~\ref{lemma2}, we particularly consider the universal gate set $\{X,CH,CCZ\}\cup\{H_k\ |\ k\in\mathbb{Z},-p(|x|)\le k\le p(|x|)\}$ with a polynomial $p(|x|)$ in the instance size $|x|$ of a given problem.
Here, $CH\equiv|0\rangle\langle 0|\otimes I+|1\rangle\langle 1|\otimes H$ is the controlled-Hadamard gate, $CCZ\equiv|0\rangle\langle 0|\otimes I^{\otimes 2}+|1\rangle\langle 1|\otimes CZ$ is the controlled-controlled-$Z$ ($CCZ$) gate, and $H_k$ is the generalized Hadamard gate such that $H_k|0\rangle=(|0\rangle+2^k|1\rangle)/\sqrt{1+4^k}$ and $H_k|1\rangle=(2^k|0\rangle-|1\rangle)/\sqrt{1+4^k}$.
Therefore, $H_0$ is the ordinary Hadamard gate $H$, and hence, from \cite{S02}, our gate set is universal.
By using our universal gate set, we can make output probabilities of any Pauli-$Z$ measurement in any polynomial-size quantum circuit $0$ or at least $2^{-q(|x|)}$ for some polynomial $q(|x|)$.
Due to this property, we can postselect any outcome for any polynomial-size quantum circuit [see Eq.~(\ref{condadpost})], which simplifies a proof of Lemma~\ref{lemma1}.
Furthermore, by using this gate set, we can perform all quantum operations required in a proof of Theorem~\ref{lemma2} without any approximation.
Note that our argument can also be applied to other universal gate sets such as $\{H,T,CZ\}$ with $T\equiv|0\rangle\langle 0|+e^{i\pi/4}|1\rangle\langle 1|$ by using the Solovay-Kitaev algorithm~\cite{DN06}.

We show the second inclusion (ii):
\begin{lemma}
\label{lemma1}
${\sf CBQP}\subseteq{\sf AdPostBQP}$.
\end{lemma}
\begin{proof}
To obtain this lemma, it is sufficient to show that for any polynomial-size linear operator $Q$ and its classical description $\mathcal{D}$, the cloning operator $C$ can be simulated in quantum polynomial time by using the {\red adaptive} postselection.
That is, our purpose is to perform the cloning operator $C$ on the input state $|\mathcal{D}\rangle$.
Let $m$ be the number of $Z$-basis projective operators included in $Q$.
By using $n$-qubit unitary operators $\{U^{(i)}\}_{i=1}^{m+1}$ and $Z$-basis projective operators $\{P^{(i)}\}_{i=1}^m$, $Q=U^{(m+1)}\prod_{i=1}^m(P^{(i)}U^{(i)})$.
We can obtain the classical description $\mathcal{D}$ of $Q$ by measuring the state $|\mathcal{D}\rangle$ in the Pauli-$Z$ basis.
The description $\mathcal{D}$ informs us about whether $P^{(i)}$ is $|0\rangle\langle 0|$ or $|1\rangle\langle 1|$ and how to construct $U^{(i)}$ from $\{X,H_k,CH,CCZ\}$ for all $i$.
Therefore, by using the postselection, we can prepare $Q|0^n\rangle$ (up to normalization) in quantum polynomial time.
When we would like to apply $U^{(i)}$, we just apply it.
On the other hand, when we apply $P^{(i)}$, we use the postselection.
Since we assume the universal gate set $\{X,H_k,CH,CCZ\}$, the postselection is possible in any case.
These efficient procedures simulate the non classically-controlled cloning operator $C$.

Next, we show that the above procedures can also be applied to simulate a classically controlled cloning operator.
Suppose that when $a\in\{0,1\}$ is $1$, we would like to apply the cloning operator $C$.
On the other hand, when $a=0$, we do not apply $C$.
Note that without loss of generality, we can assume that $C$ is controlled by a single bit $a$ because $C$ is {\red either} applied or not.
Only when $a=1$, we must apply $P^{(i)}$ to simulate the classically controlled cloning operator, which means that a classically controlled postselection seems to be necessary.
This problem can be resolved in the following way.
Let $P^{(i)}=|b\rangle\langle b|$ for $b\in\{0,1\}$.
Such classically controlled $P^{(i)}$ can be simulated by adding  an ancillary qubit $|b\rangle$ and applying the classically controlled SWAP gate as shown in Fig.~\ref{ccpost}.
Classically controlled quantum gates are allowed in ${\sf AdPostBQP}$ computation because any classically controlled quantum gate can be realized by combining elementary quantum gates in a universal gate set.

\begin{figure}[t]
\begin{center}
\includegraphics[width=12cm, clip]{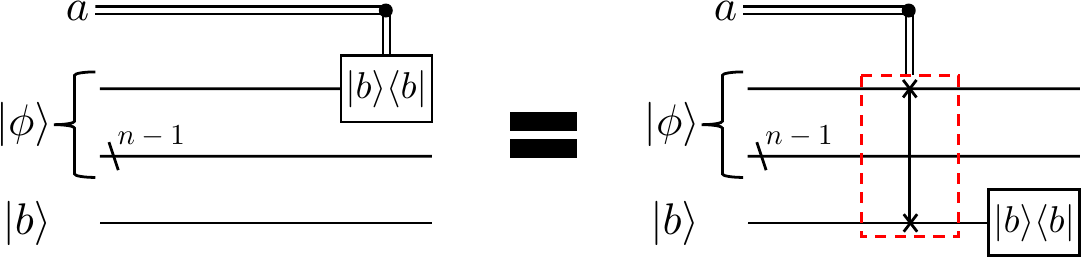}
\end{center}
\caption{Replacement of a classically controlled projector with the classically controlled SWAP gate. $|\phi\rangle$ is a quantum state immediately before applying $P^{(i)}=|b\rangle\langle b|$, where $b\in\{0,1\}$. The SWAP gate is depicted as a vertical line enclosed by a dotted red rectangle. The projector $|b\rangle\langle b|$ in the left circuit and the SWAP gate in the right circuit are applied only when $a\in\{0,1\}$ is $1$.}
\label{ccpost}
\end{figure}

In conclusion, we obtain ${\sf CBQP}\subseteq{\sf AdPostBQP}$.
\end{proof}

As the first step to obtain inclusions (iii) and (iv), we show the following theorem:
\begin{theorem}
\label{lemma2}
${\sf RwBQP}\supseteq{\sf PP}$.
\end{theorem}
\begin{proof}
We show this theorem by replacing the postselection used in the proof of ${\sf PP}\subseteq{\sf PostBQP}$ in \cite{A05} with a polynomial number of rewinding operators.
To this end, we consider the following ${\sf PP}$-complete problem~\cite{A05}: let $f:\{0,1\}^n\rightarrow\{0,1\}$ be a Boolean function computable in classical polynomial time, and $s\equiv\sum_{x\in\{0,1\}^n}f(x)$.
Decide $0<s<2^{n-1}$ or $s\ge 2^{n-1}$.
Note that it is promised that one of them is definitely satisfied.

To solve this problem with an exponentially small error probability using rewinding operators, first, we prepare 
\begin{eqnarray}
\label{prepost}
\sqrt{p}|\psi_t^\perp\rangle+\sqrt{1-p}|\psi_t\rangle,
\end{eqnarray}
where
\begin{eqnarray}
\label{p}
p&\equiv&\frac{2\alpha^2(2^n-s)^2+\beta^24^n}{2[(2^n-s)^2+s^2]}\\
|\psi_t^\perp\rangle&\equiv&\frac{\sqrt{2}\alpha(2^n-s)|0\rangle+\beta2^n|1\rangle}{\sqrt{2\alpha^2(2^n-s)^2+\beta^24^n}}\otimes |0\rangle\\
|\psi_t\rangle&\equiv&\frac{\sqrt{2}\alpha s|0\rangle+\beta(2^n-2s)|1\rangle}{\sqrt{2\alpha^2s^2+\beta^2(2^n-2s)^2}}\otimes |1\rangle
\end{eqnarray}
for positive real numbers $\alpha$ and $\beta$ such that $\alpha^2+\beta^2=1$ and $\beta/\alpha=2^k$, where $k$ is an integer whose absolute value is upper bounded by $n$.
{\red This preparation can be efficiently done without the postselection and rewinding operators by using a protocol in \cite{A05}.
First, we prepare 
\begin{eqnarray}
\label{TOCSstate1}
\frac{1}{\sqrt{2^n}}\sum_{x\in\{0,1\}^n}(H^{\otimes n}|x\rangle)|f(x)\rangle
\end{eqnarray}
in quantum polynomial time.
Then, we measure all $n$ qubits in the first register in the Pauli-$Z$ basis.
We repeat these procedures until we obtain the outcome $0^n$ or the repetition number reaches $n$.
In previous calculations~\cite{A05,AL98}, the probability of $0^n$ being output in each repetition is lower bounded by $1/4$. However, we tighten the lower bound by calculating as follows:
\begin{eqnarray}
\nonumber
&&\left[\frac{1}{\sqrt{2^n}}\sum_{y\in\{0,1\}^n}(\langle y|H^{\otimes n})\langle f(y)|\right]\left(|0^n\rangle\langle 0^n|\otimes I\right)\left[\frac{1}{\sqrt{2^n}}\sum_{x\in\{0,1\}^n}(H^{\otimes n}|x\rangle)|f(x)\rangle\right]\\
&=&\frac{1}{2^n}\sum_{x,y\in\{0,1\}^n}\frac{1}{2^n}\langle f(y)|f(x)\rangle\\
&=&\frac{M_0^2+M_1^2}{4^n}=\frac{M_0^2+(2^n-M_0)^2}{4^n}\ge\frac{1}{2},
\end{eqnarray}
where $M_b\equiv|\{x\in\{0,1\}^n: f(x)=b\}|$ for any $b\in\{0,1\}$, and we have used $M_0+M_1=2^n$.
Note that $M_1$ is equal to $s$.
Therefore, we can obtain at least one $0^n$ with probability of at least
\begin{eqnarray}
\label{prob31}
1-1/2^n.
\end{eqnarray}
When the measurement outcome is $0^n$, we obtain
\begin{eqnarray}
\label{TOCSstate2}
|\psi\rangle\equiv\cfrac{(2^n-s)|0\rangle+s|1\rangle}{\sqrt{(2^n-s)^2+s^2}}.
\end{eqnarray}
From this state, for any positive real numbers $\alpha$ and $\beta$ such that $\alpha^2+\beta^2=1$ and $\beta/\alpha=2^k$, we can prepare
\begin{eqnarray}
\label{TOCSstate3}
CH[(\alpha|0\rangle+\beta|1\rangle)|\psi\rangle]=\alpha|0\rangle|\psi\rangle+\beta|1\rangle H|\psi\rangle,
\end{eqnarray}
which is identical to the state in Eq.~(\ref{prepost}), in quantum polynomial time~\footnote{\red For clarity, we here explicitly write the procedure of Aaronson's state-preparation protocol in \cite{A05}. First, the state in Eq.~(\ref{TOCSstate1}) is prepared, and then the state in Eq.~(\ref{TOCSstate2}) is generated from it. Finally, it is transformed to the target state in Eq.~(\ref{TOCSstate3}).}.}

If the second qubit in Eq.~(\ref{prepost}) is projected onto $|1\rangle$, we can obtain
\begin{eqnarray}
\label{post}
|\phi_{\beta/\alpha}\rangle\equiv\frac{\sqrt{2}\alpha s|0\rangle+\beta(2^n-2s)|1\rangle}{\sqrt{2\alpha^2s^2+\beta^2(2^n-2s)^2}}.
\end{eqnarray}
{\red Aaronson has shown that when $0<s<2^{n-1}$, there exists an integer $k$ such that $|\langle +|\phi_{\beta/\alpha}\rangle|\ge(1+\sqrt{2})/\sqrt{6}$.
On the other hand, if $s\ge 2^{n-1}$, then $|\langle +|\phi_{\beta/\alpha}\rangle|\le1/\sqrt{2}$ holds for all $-n\le k\le n$.
Therefore, by simply measuring $n$ copies of $|\phi_{\beta/\alpha}\rangle$ in the Pauli-$X$ basis for all $k$, we can decide whether $0<s<2^{n-1}$ or $s\ge 2^{n-1}$ with an exponentially small error probability $p_{\rm err}$ in quantum polynomial time~\cite{A05}.}

However, since $p$ may be exponentially close to $1$, the efficient preparation of Eq.~(\ref{post}) is difficult without postselection.
We resolve this problem by using rewinding operators.
Our idea is to amplify the probability of $|1\rangle$ being observed by mitigating the probability of $|0\rangle$ being observed.
We propose {\red the mitigation protocol in Algorithm~\ref{mpa} (see also Fig.~\ref{miti}).}

{\red\begin{algorithm}
\begin{flushleft}
\caption{\red Mitigation protocol}
\label{mpa}
\begin{enumerate} 
\item Set $i=0$ and $c=0$.
\item By using the state in Eq.~(\ref{prepost}), prepare
\begin{eqnarray}
\label{prepare}
\sqrt{p_i}|\psi_t^\perp\rangle|+\rangle+\sqrt{1-p_i}|\psi_t\rangle|0\rangle,
\end{eqnarray}
where $p_0=p$, and measure the last register in the Pauli-$Z$ basis.
Let $z$ be the measurement outcome.
Furthermore, replace $c$ with $c+1$.
\item Depending on the values of $z$, $i$, and $c$, perform one of following steps:
\begin{enumerate}
\item[(a)] When $z=0$, replace $i$ with $i+1$, reset $c$ to $0$, and obtain
\begin{eqnarray}
\label{i+1}
\sqrt{p_{i+1}}|\psi_t^\perp\rangle+\sqrt{1-p_{i+1}}|\psi_t\rangle,
\end{eqnarray}
where
\begin{eqnarray}
\label{suc}
p_{i+1}=\frac{p_i}{2-p_i}.
\end{eqnarray}
If $i+1<2n+3$, do step 2 by using the state in Eq.~(\ref{i+1}).
On the other hand, if $i+1=2n+3$, output the state in Eq.~(\ref{i+1}) and halt the mitigation protocol.
\item[(b)] When $z=1$ and $c<3n$, apply the rewinding operator $R$ and do step 2 again for the same $i$.
\item[(c)] When $z=1$ and $c=3n$, answer $0<s<2^{n-1}$ or $s\ge2^{n-1}$ uniformly at random, and halt the mitigation protocol.
\end{enumerate}
\end{enumerate}
\end{flushleft}
\end{algorithm}}

\begin{figure}[t]
\begin{center}
\includegraphics[width=12cm, clip]{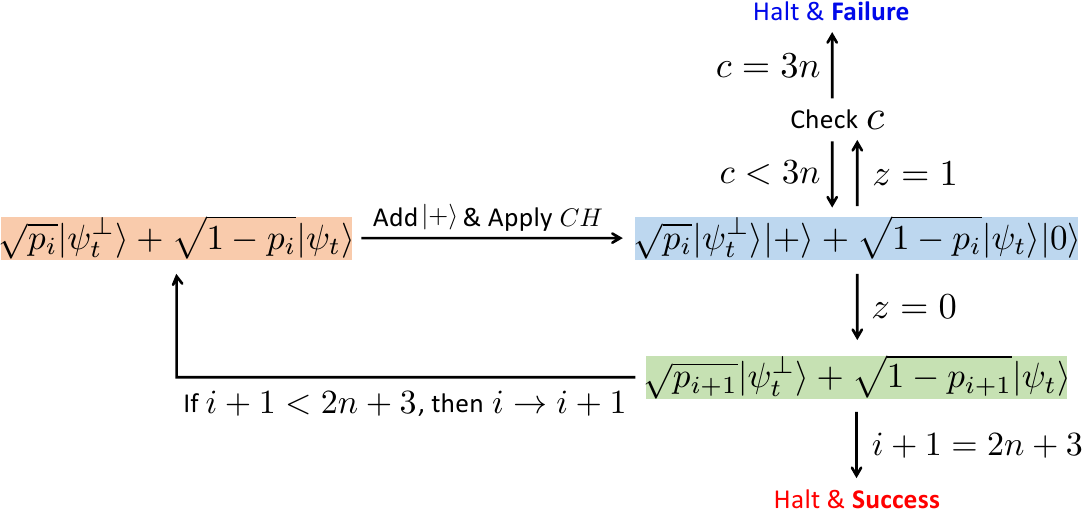}
\end{center}
\caption{Schematic diagram of our mitigation protocol. {\bf Failure} and {\bf Success} indicate step (c) and the case that the state in Eq.~(\ref{ideal}) is output, respectively.}
\label{miti}
\end{figure}

In this protocol, $i$ and $c$ count how many times the mitigation succeeds and how many times the mitigation fails for a single $i$, respectively.
From Eq.~(\ref{suc}),
\begin{eqnarray}
\label{suc2}
\frac{1-p_{i+1}}{p_{i+1}}=2\frac{1-p_i}{p_i},
\end{eqnarray}
and hence we succeed in mitigating the amplitude of the nontarget state $|\psi_t^\perp\rangle$.
{\red Note that Eq.~(\ref{suc}) is derived from Eq.~(\ref{projectionTOCS}) and the fact that the normalization factor of Eq.~(\ref{projectionTOCS}) is $\sqrt{1-p_i/2}$.}

To evaluate our mitigation protocol, we show the following Claim in Appendix~\ref{D}:
\begin{claim}
The mitigation protocol succeeds, i.e., the counter $i$ becomes $2n+3$ with probability of at least
\begin{eqnarray}
\label{success2}
1-\frac{5n}{8^n}.
\end{eqnarray}
In this case, the output state satisfies
\begin{eqnarray}
\label{mitigate2}
\frac{1-p_{2n+3}}{p_{2n+3}}\ge 1.
\end{eqnarray}
\end{claim}
From Eq.~(\ref{success2}),
\begin{eqnarray}
\label{ideal}
\sqrt{p_{2n+3}}|\psi_t^\perp\rangle+\sqrt{1-p_{2n+3}}|\psi_t\rangle
\end{eqnarray}
is output with probability of at least $1-5n/8^n$.
From Eq.~(\ref{mitigate2}), we can obtain the outcome $1$ with probability of at least $1/2$ by measuring the second qubit in Eq.~(\ref{ideal}).
If we obtain $0$, we do the same measurement again by using the rewinding operator.
Therefore, by repeating this procedure $n$ times, we obtain the outcome $1$ with probability of at least
\begin{eqnarray}
\label{prob32}
1-1/2^n.
\end{eqnarray}
In total, from Eqs.~(\ref{prob31}), (\ref{success2}), and (\ref{prob32}) and the fact that the mitigation protocol is used for all $-n\le k\le n$, with probability of at least
\begin{eqnarray}
p_{\rm suc}\equiv\left[\left(1-\frac{1}{2^n}\right)^2\left(1-\frac{5n}{8^n}\right)\right]^{n(2n+1)},
\end{eqnarray}
we obtain $n$ copies of $|\phi_{\beta/\alpha}\rangle$ for all $k$.
As a result, we can correctly decide whether $0<s<2^{n-1}$ or $s\ge 2^{n-1}$ in polynomial time with probability of at least $p_{\rm suc}(1-p_{\rm err})$ that is exponentially close to $1$.
\end{proof}

From Theorem~\ref{lemma2}, we obtain the following corollary:
\begin{corollary}
\label{nlemma4}
${\sf BQP}^{\sf PP}_{{\rm classical}}\subseteq{\sf AdPostBQP}\subseteq{\sf RwBQP}$.
\end{corollary}
\begin{proof}
From Lemmas~\ref{nlemma1} and \ref{lemma1}, it is sufficient to show ${\sf BQP}^{\sf PP}_{{\rm classical}}\subseteq{\sf RwBQP}$ and ${\sf AdPostBQP}\subseteq{\sf RwBQP}$ to obtain this corollary.
First, we show the former inclusion.
From Def.~\ref{BRQP}, it is obvious that any process in ${\sf BQP}$ can be simulated by a process in ${\sf RwBQP}$ in polynomial time.
Furthermore, from Corollary~\ref{corollary2} and Theorem~\ref{lemma2}, the ${\sf PP}$ oracle can be replaced with the ${\sf RwBQP}$ oracle.
Therefore, from Corollary~\ref{corollary3}, ${\sf BQP}^{\sf PP}_{{\rm classical}}\subseteq{\sf RwBQP}^{\sf RwBQP}_{{\rm classical}}={\sf RwBQP}$.

To obtain the latter inclusion, we show that each postselection can be simulated by rewinding operators.
From Def.~\ref{AdPostBQP}, each postselection (i.e., each projector) acts on a single qubit.
Therefore, we can write a quantum state immediately before a postselection as $\alpha|0\rangle_p|\psi_0\rangle+\beta|1\rangle_p|\psi_1\rangle$ for some quantum states $|\psi_0\rangle$ and $|\psi_1\rangle$ and complex numbers $\alpha$ and $\beta$ such that $|\alpha|^2+|\beta|^2=1$.
Here, the subscript $p$ denotes the postselection register.
Although $|\beta|^2$ may be exponentially small, the postselection onto $|1\rangle$ can be simulated by using our mitigation protocol.
\end{proof}

As the most important difference between the proof of Theorem~\ref{lemma2} and that of ${\sf PostBQP}\supseteq{\sf PP}$ in \cite{A05}, we have not used postselections of outputs whose probabilities are exponentially small by proposing the mitigation protocol.
To state the difference more explicitly on the technical level, we show the following corollary:
\begin{corollary}
\label{main1}
For any polynomial function {\red$P(|x|)$} in the size $|x|$ of an instance $x$, ${\sf PP}={\sf PostBQP}$ holds even if only non-adaptive postselections of outputs whose probabilities are {\red $1-O(1/P(|x|))$} are allowed.
\end{corollary}

\begin{proof}
We can obtain this corollary by slightly modifying our mitigation protocol and showing that it can be realized with non-adaptive postselections of outputs whose probabilities are at least $q\equiv1-{\red 1/P(|x|)}$.
For some natural number $m$, we prepare
\begin{eqnarray}
\sqrt{p}|\psi_t^\perp\rangle\left(\sqrt{q}|0\rangle+\sqrt{1-q}|1\rangle\right)^{\otimes m}+\sqrt{1-p}|\psi_t\rangle|0\rangle^{\otimes m}
\end{eqnarray}
instead of the state in Eq.~(\ref{prepare}).
By postselecting $m$ qubits in the second register onto $|0\rangle$ one by one, we obtain
\begin{eqnarray}
\label{new24}
\frac{\sqrt{pq^m}|\psi_t^\perp\rangle+\sqrt{1-p}|\psi_t\rangle}{\sqrt{1-p+pq^m}}.
\end{eqnarray}
These $m$ postselections are non-adaptive ones of outputs whose probabilities are at least $q$.
If the amplitude of $|\psi_t\rangle$ in Eq.~(\ref{new24}) is at least $\sqrt{1/2}$, we can obtain $|\psi_t\rangle$ with at least a constant probability, and hence we can solve the ${\sf PP}$-complete problem.
Such the amplitude is realized by setting $m\ge\log_2{[p/(1-p)]}/\log_2{(1/q)}$.
Since $\log_2{[p/(1-p)]}/\log_2{(1/q)}\le2(n+1)/\log_2{(1/q)}\le2(n+1)O({\red P(|x|)})$ from {\red Eq.~(\ref{ubTOCS}) in} Appendix~\ref{D}, where $n$ is at most a polynomial function in $|x|$, a polynomial number of postselections are sufficient in the above argument.
\end{proof}
 
\section{Restricted Rewindable Quantum Computation}
{\red\subsection{\red Power of Single Rewinding Operator}}
In Sec.~\ref{IIIA}, a polynomial number of rewinding operators was available.
If the number is restricted to a constant, the question is: how is the rewinding useful for universal quantum computation?
We show that a single rewinding operator is sufficient to solve the following problem with a constant probability, which seems hard for universal quantum computation:
\begin{definition}[Collision-finding Problem]
\label{problem}
Given the function family $\mathcal{F}\equiv\{f_K\}_{K\in\mathcal{K}}$ in Theorem~\ref{crf} and a parameter $K$ for $\mathcal{F}$, output a pair $(x,x')$ with $x,x'\in\mathbb{Z}_q^n\times \chi^m\times\{0,1\}$ such that (i) $x\neq x'$ and (ii) $f_K(x)=f_K(x')$.
\end{definition}

\begin{theorem}
\label{easy}
Assume that a rewinding operator can be applied in one step, and there is no polynomial-time quantum algorithm solving ${\rm SIVP}_{p(n)}$ for some polynomial $p(n)$ in $n$.
Then, the problem in Def.~\ref{problem} can be solved with probability of at least $\delta/2(1-o(1))$ by uniformly generated polynomial-size quantum circuits with a single rewinding operator, but it cannot be achieved without rewinding operators.
Here, the probability is taken over the uniform distribution on $\mathcal{K}$ and the randomness used in a quantum circuit to solve the problem.
\end{theorem}

The sketch of the proof is explained in Sec.~\ref{IB}, and the rigorous proof is given in Appendix~\ref{G3}.
Here, we explain the implication of Theorem~\ref{easy}.
In Sec.~\ref{IIIA}, we have shown the equivalence between the postselection and rewinding.
In contrast, Theorem~\ref{easy} may represent their difference.
A possible approach to solving the problem in Def.~\ref{problem} is to generate two copies of the state in Eq.~(\ref{preimage}) (more precisely, Eq.~(\ref{two}) in Appendix~\ref{G3}) by using the postselection.
As a straightforward way, this can be achieved by postselecting the second register in the state $\sum_x|x\rangle|f_K(x)\rangle$ (more precisely, Eq.~(\ref{firststate}) in Appendix~\ref{G3}) onto the same $f_K(x)$.
However, it requires the postselection of a polynomial number of qubits (or the postselection of states whose amplitudes are exponentially small), while a single qubit is sufficient for the rewinding.
Furthermore, since we do not know which $f_K(x)$ has exactly two different preimages, the postselection applied to the second copy of $\sum_x|x\rangle|f_K(x)\rangle$ needs to be adaptive, i.e., it depends on $f_K(x)$ obtained from the first copy.
On the other hand, a non-adaptive (i.e., non classically-controlled) rewinding operator is sufficient for solving the problem (with a constant probability).
Although there may be other ways to solve the problem by using a non-adaptive postselection of a single qubit, the above discussion may imply that the rewinding is superior to the postselection in some situations where the number of qubits to be rewound or postselected is restricted, and copies (i.e., $(\sum_x|x\rangle|f_K(x)\rangle)^{\otimes 2}$) are not processed collectively.

We also show a superiority of a single rewinding operator under a different assumption.
To this end, we use the statistical difference (SD) problem, which is ${\sf SZK}$-complete~\cite{SV03}, and show the following theorem:
\begin{definition}[Statistical Difference Problem~\cite{SV03}]
\label{SDP}
Given classical descriptions of two Boolean circuits $C_0,C_1:\{0,1\}^n\rightarrow\{0,1\}^m$ with natural numbers $n$ and $m$, let $P_0$ and $P_1$ be distributions of $C_0(x)$ and $C_1(x)$ with uniformly random inputs $x\in\{0,1\}^n$, respectively.
Decide whether $D_{\rm TV}(P_0,P_1)<2^{-O(n^c)}$ or $D_{\rm TV}(P_0,P_1)>1-2^{-O(n^c)}$ for some positive constant $c$, where $D_{\rm TV}(\cdot,\cdot)$ is the total variation distance.
\end{definition}

\begin{theorem}
\label{newlemma}
The SD problem defined in Def.~\ref{SDP} is in ${\sf RwBQP}(1/2-2^{-O(n^c)},2\cdot2^{-O(n^c)})(1)$, where $n$ and $c$ are the problem size and some positive constant as defined in Def.~\ref{SDP}.
\end{theorem}
The proof is given in Appendix~\ref{newH}.
From Theorem~\ref{newlemma}, we obtain the following corollary:
\begin{corollary}
\label{new23}
Let ${\sf RwBQP}(1)\equiv\bigcup_{1/(c-s)\in{\rm poly}(|x|)}{\sf RwBQP}(c,s)(1)$ for the set ${\rm poly}(|x|)$ of all polynomial functions in the size $|x|$ of an instance $x$.
Then, ${\sf RwBQP}(1)\supset{\sf BQP}$ unless ${\sf BQP}\supseteq{\sf SZK}$.
\end{corollary}
\begin{proof}
From Theorem~\ref{newlemma}, if ${\sf RwBQP}(1)\subseteq{\sf BQP}$, then ${\sf SZK}\subseteq{\sf BQP}$.
\end{proof}

For example, by assuming that the decision version of SIVP, gapSIVP, is hard for universal quantum computation, Corollary~\ref{new23} implies that a single rewinding operator is sufficient to achieve a task that is intractable for universal quantum computation.
This is because the gapSIVP (with an appropriate parameter) is in ${\sf SZK}$~\cite{PV08}.
Therefore, Corollary~\ref{new23} shows the superiority of a single rewinding operator for promise problems, while Theorem~\ref{easy} shows it for the search problem.

As a simple observation, a single cloning operator should also be sufficient to exceed universal quantum computation.
This is because any problem in ${\sf PostBQP}$ is efficiently solvable with a single cloning operator.
By applying the cloning operator $C$ on a classical description of the output state
\begin{eqnarray}
\label{CPostBQP}
\frac{\left(I\otimes |1\rangle\langle 1|\otimes I^{\otimes n-2}\right)U_x|0^n\rangle}{\sqrt{{\rm Pr}[p=1]}}
\end{eqnarray}
of a ${\sf PostBQP}$ computation (see Def.~\ref{postbqpa}), the state in Eq.~(\ref{CPostBQP}) is prepared in one step.
Then, by measuring its first qubit in the computational basis, we can solve the ${\sf PostBQP}$ problem.

{\red\subsection{Rewindable Clifford Circuits}
\label{A}}
{\red As another observation}, we consider how useful rewinding operators are for restricted classes of quantum computation.
Particularly, {\red in this subsection,} we focus on Clifford circuits defined as follows:
\begin{definition}[Clifford Circuits]
\label{clifford}
Let $n$ be any natural number.
An $n$-qubit Clifford circuit $C_n$ is a quantum circuit such that
\begin{enumerate}
\item The initial state is $|0^n\rangle$.
\item An applied $n$-qubit unitary operator $U$ consists of elementary quantum gates chosen from $\{H,S,CZ\}$, where $H$ is the Hadamard gate, $S\equiv|0\rangle\langle 0|+i|1\rangle\langle 1|$, and $CZ$ is the $CZ$ gate.
\item All or a part of qubits in the final state $U|0^n\rangle$ are measured in the Pauli-$Z$ basis $\{|0\rangle,|1\rangle\}$.
\end{enumerate}
Particularly when $U$ consists of a polynomial number of elementary quantum gates in $n$, we say that the size of $C_n$ is the polynomial.
\end{definition}
In Definition~\ref{clifford}, the unitary operator $U$ is in the Clifford group that is defined as the group of unitary operators normalizing the Pauli group.
Here, $n$-qubit Pauli group ${\bf P}_n$ is defined as ${\bf P}_n\equiv\{e^{ik\pi/2}\bigotimes_{i=1}^n\sigma^{(j_i)}\ |\ k,j_i\in\{0,1,2,3\}\ {\rm for}\ 1\le i\le n\}$, where $\sigma^{(0)}=I$, $\sigma^{(1)}{\red=}X$, $\sigma^{(2)}=Z$, and $\sigma^{(3)}=iXZ$.
Therefore, we call $C_n$ a Clifford circuit.

It is well known that Clifford circuits are classically simulatable.
More formally, the Gottesman-Knill theorem~\cite{G99} states as follows:
\begin{theorem}
\label{GK}
Let $n$ be any natural number and $P_{S}$ be any rank-one $Z$-basis projector applied on qubits in the subset $S$ of $n$ qubits.
Any $n$-qubit polynomial-size Clifford circuit $C_n$ can be strongly simulated in classical polynomial time in $n$.
That is, for any $S$ and Clifford unitary $U$ in $C_n$, the output probability $\langle 0^n|U^\dag P_SU|0^n\rangle$ can be exactly computed in classical polynomial time in $n$.
\end{theorem}
Simply speaking, Theorem~\ref{GK} has been shown by using the stabilizer formalism.
The initial state $|0^n\rangle$ is a unique common $+1$ eigenstate of $\langle\{Z_i\}_{i=1}^n\rangle$, where $Z_i$ is the Pauli-$Z$ operator applied on the $i$th qubit.
Here, $\langle\cdot\rangle$ means that $\{Z_i\}_{i=1}^n$ are $n$ independent generators of a stabilizer group. 
The time evolution from $|0^n\rangle$ to $u_1|0^n\rangle$ with an elementary gate $u_1$ can be simulated by replacing $\langle\{Z_i\}_{i=1}^n\rangle$ with $\langle\{u_1Z_iu_1^\dag\}_{i=1}^n\rangle$.
Since stabilizer groups are abelian subgroups of the Pauli group and $u_1$ is in the Clifford group, this replacement can be done in classical polynomial time.
By repeating a similar procedure for other elementary gates, we can uniquely represent $U|0^n\rangle$ by $\langle\{UZ_iU^\dag\}_{i=1}^n\rangle$.
When the outcome of the Pauli-$Z$ measurement on the $i$th qubit is $o_i\in\{0,1\}$, the measurement can be simulated by replacing a generator\footnote{Since the stabilizer group is a subgroup of the Pauli group, each generator commutes or anticommutes with $Z_i$. Furthermore, if there exist two or more generators that anticommute with $Z_i$, we can reduce the number to one by multiplying them.} that anticommutes with $Z_i$ with $(-1)^{o_i}Z_i$.
This replacement can also be done in classical polynomial time.
Furthermore, it can be easily shown that the probability of obtaining $o_i$ is $1/2$ for any $o_i\in\{0,1\}$.
Note that if all generators commute with $Z_i$, the measurement does not change the generators.
In this case, the stabilizer group certainly contains $Z_i$ or $-Z_i$, and hence the measurement outcome is definitely $0$ or $1$, respectively.
It is possible to check which, $Z_i$ or $-Z_i$, is contained in the stabilizer group in classical polynomial time.
In short, since $n$ generators of any stabilizer group can be represented as a polynomial-size bit string $y$, Clifford circuits can be efficiently simulated by classically processing $y$.

We show that this classical simulatability holds even if rewinding operators are supplied:
\begin{theorem}
\label{cliffordrew}
Let $n$ be any natural number, $P_S$ be any rank-one $Z$-basis projector applied on qubits in the subset $S$ of $n$ qubits, and $Q_{\rm Clifford}$ be any $n$-qubit operator composed of a polynomial number of $\{H,S,CZ\}$, single-qubit measurements in the computational basis, and the (classically controlled) rewinding operator $R$.
Let $Q_{\rm Clifford}^{(z)}$ be the same as $Q_{\rm Clifford}$ except for that the $i$th measurement is replaced with $|z_i\rangle\langle z_i|$ for all $i$, and $z_i$ is the $i$th bit of $z$.
For any $z$ and $S$, the output probability $\langle 0^n|Q_{\rm Clifford}^{(z)\dag} P_SQ_{\rm Clifford}^{(z)}|0^n\rangle$ can be exactly computed in classical polynomial time in $n$.
\end{theorem}
\begin{proof}
From the stabilizer formalism explained above, we can simulate Clifford unitary operators and Pauli-$Z$ measurements with outcomes $\{z_i\}_i$ in classical polynomial time by processing polynomial-length bit strings.
Therefore, the remaining task is to show that the rewinding operator applied on a quantum state generated by Clifford unitary operators and Pauli-$Z$ measurements can also be simulated in classical polynomial time.

Suppose that we now have an $n$-qubit state $Q|0^n\rangle$ (up to normalization), where $Q$ is an $n$-qubit linear operator composed of a polynomial number of Clifford unitary operators and Pauli-$Z$ measurements.
Then, we measure its first qubit in the $Z$ basis and obtain the post-measurement state $(|z\rangle\langle z|\otimes I^{\otimes n-1})Q|0^n\rangle$ (up to normalization) with $z\in\{0,1\}$.
Note that since the SWAP gate is a Clifford unitary, without loss of generality, we can assume that the first qubit is measured.
Let $a\in\{0,1\}$.
Finally, if $z=a$, we apply the rewinding operator $R$ on the post-measurement state.
Otherwise, we do nothing.
For any $z,a\in\{0,1\}$, the above procedures can be classically simulated in the following way:
\begin{enumerate}
\item By using the stabilizer formalism, represent $Q|0^n\rangle$ by a polynomial-length bit string $y$.
\item From $y$, make a duplicate $y'$.
\item Simulate the Pauli-$Z$ measurement on $Q|0^n\rangle$ by replacing $y$ with another polynomial-length bit string $y_z$.
\item When $z=a$, replace $y_z$ with the duplicate $y'$.
Otherwise, do nothing.
\end{enumerate}
The above argument can be used even if the rewinding operator is conditioned on a number of bits or is not classically controlled.
\end{proof}

Since Clifford circuits are classically simulatable even under postselection, Theorem~\ref{cliffordrew} also implies an equivalence between the postselection and rewinding.

{\red\subsection{Rewindable IQP Circuits}
\label{B}}
In previous sections, we have considered the effect of the rewinding on universal quantum computation and classically simulatable quantum computation.
How about sub-universal quantum computing models that are (believed to be) neither universal nor classically simulatable?
Particularly, we consider IQP circuits~\cite{BJS11} as a sub-universal quantum computing model.
They are defined by replacing Clifford unitary operators $U$ in Clifford circuits with unitary operators $E$ diagonalizable in the Pauli-$X$ basis $\{|+\rangle,|-\rangle\}$:
\begin{definition}[IQP Circuits~\cite{BJS11}]
\label{IQP}
Let $n$ be any natural number.
$n$-qubit IQP circuits are quantum circuits such that
\begin{enumerate}
\item The initial state is $|0^n\rangle$.
\item An applied $n$-qubit unitary operator is $E\equiv H^{\otimes n}DH^{\otimes n}$, where D is a unitary operator diagonalizable in the Pauli-$Z$ basis $\{|0\rangle,|1\rangle\}$.
\item All or some of the qubits in the final state $E|0^n\rangle$ are measured in the Pauli-$Z$ basis.
\end{enumerate}
\end{definition}
In particular, we consider the case that $D$ consists of a polynomial number of elementary gates chosen from $\{R_z(\theta),CZ\ |\ 0\le\theta<2\pi\}$, where $R_z(\theta)\equiv|0\rangle\langle0|+e^{i\theta}|1\rangle\langle 1|$.
Since all elementary gates commute with each other, we can represent $D$ by $(\bigotimes_{i=1}^nR_z(\theta_i))D_{CZ}$, where $D_{CZ}$ is a unitary operator composed of only $CZ$'s.

When $D_{CZ}$ satisfies some property mentioned later, the IQP circuits can be regarded as measurement-based quantum computation (MBQC)~\cite{RB01}.
To define MBQC, we first introduce the brickwork state as follows:
\begin{definition}[Brickwork State~\cite{BFK09}]
Let $n$ and $m$ be natural numbers.
A brickwork state $|B\rangle$ is an $nm$-qubit entangled state with $m\equiv 5\ ({\rm mod}\ 8)$ that is constructed as follows:
\begin{enumerate}
\item Prepare $|+\rangle^{\otimes nm}$ and assign an index $(i,j)$ to each qubit, where $|+\rangle\equiv(|0\rangle+|1\rangle)/\sqrt{2}$, $1\le i\le n$, and $1\le j\le m$.
\item For any $1\le i\le n$ and $1\le j\le m-1$, apply $CZ$ to two qubits labeled by $(i,j)$ and $(i,j+1)$.
\item For any odd $i$ and $j\equiv 3\ ({\rm mod}\ 8)$, apply $CZ$ to two qubits labeled by $(i,j)$ and $(i+1,j)$ or $(i,j+2)$ and $(i+1,j+2)$.
\item For any even $i$ and $j\equiv 7\ ({\rm mod}\ 8)$, apply $CZ$ to two qubits labeled by $(i,j)$ and $(i+1,j)$ or $(i,j+2)$ and $(i+1,j+2)$.
\end{enumerate}
\end{definition}
We define MBQC by using the brickwork state $|B\rangle$ as follows:
\begin{definition}[MBQC]
\label{MBQC}
Let $m\equiv 5n\ ({\rm mod}\ 8)$ for any natural number $n$.
MBQC is a quantum computing model proceeding as follows:
\begin{enumerate}
\item Prepare the $m$-qubit brickwork state $|B\rangle$.
\item Measure each qubit in $|B\rangle$ one by one in the basis $\{|+_\theta\rangle,|+_{\theta+\pi}\rangle\}$, where $|+_\theta\rangle\equiv(|0\rangle+e^{i\theta}|1\rangle)/\sqrt{2}$ for any $\theta\in\mathbb{R}$.
Each measurement basis (i.e., the value of $\theta$) is calculated from all the previous measurement outcomes in classical polynomial time.
\end{enumerate}
\end{definition}
The universality of this computing model has been shown in \cite{BFK09}.

From Defs.~\ref{IQP} and \ref{MBQC}, we notice that when $D_{CZ}(|+\rangle^{\otimes n})=|B\rangle$, the measurement basis for the $i$th qubit is $\{|+_{\theta_i}\rangle,|+_{\theta_i+\pi}\rangle\}$, and all measurement bases do not depend on previous measurement outcomes, MBQC becomes an IQP circuit.
In short, (important subclasses of) IQP circuits are MBQC with non-adaptive single-qubit measurements.

Suppose that we would like to perform quantum computation by measuring all qubits in some $n$-qubit state $|\psi\rangle$ in the Pauli-$Z$ basis.
To simulate this quantum computation by MBQC, we prepare the $m\ (\ge n)$-qubit brickwork state $|B\rangle$ whose qubits can be divided into two sets, $M$ and $O$, each of which includes $(m-n)$ and $n$ qubits, respectively.
By measuring all qubits in the set $M$ in appropriate bases, the state of $n$ qubits in the set $O$ becomes $H^{\otimes n}|\psi\rangle$ (up to a byproduct operator\footnote{Since byproduct operators are tensor products of Pauli operators, they can be efficiently removed by classical postprocessing.}).
Therefore, by measuring all qubits in $O$ in the Pauli-$X$ basis, we can simulate the target quantum computation.
As an important point, for any single-qubit measurement on any qubit in the set $M$, the probability of the outcome being $0$ (corresponding to $|+_\theta\rangle\langle +_\theta|$) is exactly $1/2$, while it may not be the case for qubits in the set $O$.
Furthermore, when all outcomes of qubits in the set $M$ are $0$, adaptive measurements are not necessary.
These properties are immediately obtained from the gate teleportation (see e.g., Fig. 1 in \cite{MDF17}) and will be used to show that if rewinding operators are given, then universal quantum computation is possible by MBQC with non-adaptive single-qubit measurements.

We show that polynomial-size IQP circuits can solve any ${\sf PP}$ problem if a polynomial number of rewinding operators are supplied.
To formally state this, we define rewindable IQP circuits as follows:
\begin{definition}[Rewindable IQP Circuits]
\label{rewIQP}
Let $n$ be any natural number.
$n$-qubit rewindable IQP circuits are quantum circuits such that
\begin{enumerate}
\item The initial state is $|0^n\rangle$.
\item An applied $n$-qubit unitary operator is $E\equiv H^{\otimes n}DH^{\otimes n}$, where D is a unitary operator diagonalizable in the Pauli-$Z$ basis $\{|0\rangle,|1\rangle\}$.
\item All or some of the qubits in the final state $E|0^n\rangle$ are measured in the Pauli-$Z$ basis.
For each qubit, if the outcome\footnote{Without loss of generality, we can assume that an undesirable outcome is $1$ because the bit-flip operation can be performed in IQP circuits.} is $1$, it is allowed to apply a rewinding operator and then perform the measurement again and again.
\end{enumerate}
Particularly when $D$ consists of a polynomial number of elementary gates diagonalizable in the Pauli-$Z$ basis and the number of uses of rewinding operators grows at most polynomially with $n$, we say that $n$-qubit rewindable IQP circuits are polynomial size.
\end{definition}

By using Def.~\ref{rewIQP}, our result is formalized as follows:
\begin{theorem}
\label{rewindIQP}
Let $L=(L_{\rm yes},L_{\rm no})\subseteq\{0,1\}^\ast$ be a promise problem in ${\sf PP}$.
Then, there exists a uniform family of polynomial-size rewindable IQP circuits that decides $x\in L_{\rm yes}$ or $x\in L_{\rm no}$ with an exponentially small error probability for a given instance $x$.
\end{theorem}
\begin{proof}
First, we show that polynomial-size rewindable IQP circuits can prepare output states of polynomial-size quantum circuits in any uniform family.
To this end, we use the relation between IQP circuits and MBQC with non-adaptive single-qubit measurements explained above.
Due to the relation, our purpose is to show that if a polynomial number of rewinding operators are given, then we can perform universal quantum computation (with an exponentially small error) by MBQC with non-adaptive single-qubit measurements.
Let $|\psi\rangle$ be any $n$-qubit quantum state that can be generated in quantum polynomial time.
From \cite{BFK09}, $|\psi\rangle$ can also be generated in quantum polynomial time by using the $m\ (\ge n)$-qubit brickwork state (see Def.~\ref{MBQC}) with $m$ being a polynomial in $n$. 
As explained above, $m$ qubits can be divided into a set $M$ of $(m-n)$ qubits and a set $O$ of the remaining $n$ qubits.
When all measurement outcomes in the set $M$ are $0$, the $n$-qubit state in the set $O$ exactly becomes $|\psi\rangle$ without adaptive measurements.
This implies that to achieve our purpose, it is sufficient to show that we can make all measurement outcomes in the set $M$ $0$'s with probability exponentially close to $1$ by using a polynomial number of rewinding operators.
Due to the property that the probability of $0$ being output is exactly $1/2$ for the measurements\footnote{From Def.~\ref{MBQC}, we consider only single-qubit measurements in the $x$-$y$ plane of the Bloch sphere.} of all qubits in $M$, we can make it $1-2^{-(m-n)}$ in the following way:
\begin{enumerate}
\item Let $c=0$.
\item Suppose that we now have $|\psi\rangle$.
Measure its first qubit (in $M$) in the basis
\begin{eqnarray}
\{|+_\theta\rangle\langle +_\theta |,|+_{\theta+\pi}\rangle\langle +_{\theta+\pi}|\}
\end{eqnarray}
with $\theta\in\mathbb{R}$.
\item Depending on the measurement outcome and the value of $c$, perform one of following steps:
\begin{enumerate}
\item[{\red(a)}] If the measurement outcome is $0$ (corresponding to $|+_\theta\rangle\langle +_\theta |$), then halt the protocol.
\item[{\red(b)}] If it is $1$ (corresponding to $|+_{\theta+\pi}\rangle\langle +_{\theta+\pi}|$) and $c<m-n$, replace $c$ with $c+1$, apply the rewinding operator $R$ on the post-measurement state, and perform step $2$ by replacing $\{|+_\theta\rangle\langle +_\theta |,|+_{\theta+\pi}\rangle\langle +_{\theta+\pi}|\}$ with the $Z$ basis.
\item[{\red(c)}] Otherwise, halt the protocol.
\end{enumerate}
\end{enumerate}
The reason we replace the measurement basis in step 2 with the $Z$ basis in step 3(b) comes from Def.~\ref{Rewinding}.
Since our rewinding operator works for only Pauli-$Z$ measurements, when we choose the measurement basis $\{|+_\theta\rangle\langle +_\theta |,|+_{\theta+\pi}\rangle\langle +_{\theta+\pi}|\}$ in step 2, the state after the rewinding is
\begin{eqnarray}
\label{rotate}
H_1R_{z,1}(-\theta)|\psi\rangle,
\end{eqnarray}
where $R_z(\phi)\equiv|0\rangle\langle 0|+e^{i\phi}|1\rangle\langle1|$ for any $\phi\in\mathbb{R}$, and the subscript number represents the qubit on which the gate is applied.
Therefore, the measurement of $|\psi\rangle$ in the basis $\{|+_\theta\rangle\langle +_\theta |,|+_{\theta+\pi}\rangle\langle +_{\theta+\pi}|\}$ can be simulated by measuring the state in Eq.~(\ref{rotate}) in the $Z$ basis.

In the case of step $3(a)$, it is obvious that our purpose is achieved, i.e., the above protocol succeeds.
On the other hand, in the case of step $3(c)$, the protocol fails to obtain the outcome $0$, but it happens with the exponentially small probability $2^{-(m-n)}$.
 Therefore, the probability of the above protocol succeeding for all qubits in $M$ is $(1-2^{-(m-n)})^{m-n}$, which is exponentially close to one.

Second, by using the above argument, for any instance $x$ of the ${\sf PostBQP}$ problem, we can prepare
\begin{eqnarray}
\label{IQPPP}
\left(\prod_{j=1}^qCH_{2,n+j}\right)\left[U_x|0^n\rangle\otimes\left(\bigotimes_{i=n+1}^{n+q}|+\rangle_i\right)\right]
\end{eqnarray}
as an output state of the polynomial-size rewindable IQP circuit.
Here, $CH_{2,n+j}$ is the controlled-Hadamard gate whose control and target qubits are the second and $(n+j)$th ones, respectively, for $1\le j\le q$.
For the definitions of $U_x$, $n$, and $q$, see Def.~\ref{postbqpa}.
For the convenience, let $U_x|0^n\rangle=\alpha|0\rangle_2|\psi_0\rangle+\beta|1\rangle_2|\psi_1\rangle$ for some $(n-1)$-qubit states $|\psi_0\rangle$ and $|\psi_1\rangle$ and complex numbers $\alpha$ and $\beta$ such that $|\alpha|^2+|\beta|^2=1$.
According to our mitigation protocol, when the last $q$ qubits in Eq.~(\ref{IQPPP}) are measured in the computational basis, and the measurement outcome is $0^q$, we obtain $|\phi\rangle=\alpha'|0\rangle_2|\psi_0\rangle+\beta'|1\rangle_2|\psi_1\rangle$ such that $\beta'/\alpha'={\red\sqrt{2^q}}\beta/\alpha$ (see also Fig.~\ref{mitigationpara}).
From Appendix~\ref{D}, we can efficiently obtain the outcome $0^q$ with probability exponentially close to one by using a polynomial number of rewinding operators.
From $|\alpha|^2+|\beta|^2=1$ and Def.~\ref{postbqpa}, ${\red|\beta'/\alpha'|}\ge1$, and hence the postselection onto $|1\rangle_2$ succeeds with probability of at least $1/2$.
By using the rewinding operators again, this postselection probability is increased to $1-o(1)$.
Finally, the ${\sf PostBQP}$ problem is solved by measuring the first qubit of $|\psi_1\rangle$ in the computational basis.
The proof is completed from ${\sf PostBQP}={\sf PP}$ and the fact that the probability amplification is possible in {\sf PostBQP}~\cite{A05}.
\end{proof}

 \begin{figure}[t]
\begin{center}
\includegraphics[width=16.5cm, clip]{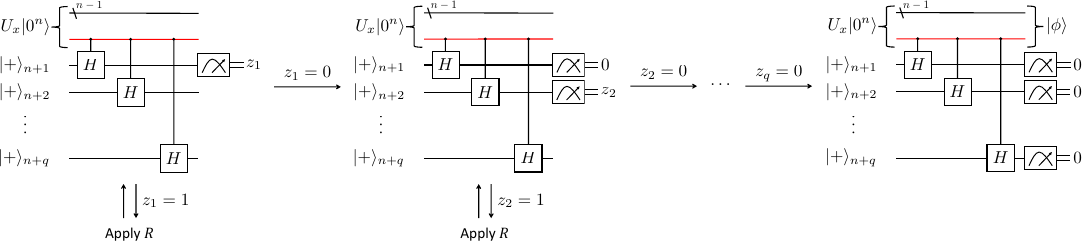}
\end{center}
\caption{Quantum circuit diagram of our mitigation protocol.
$R$ represents the rewinding operator, and $|\phi\rangle=\alpha'|0\rangle_2|\psi_0\rangle+\beta'|1\rangle_2|\psi_1\rangle$.
The red wire represents the second qubit of $U_x|0^n\rangle$ that will be postselected.}
\label{mitigationpara}
\end{figure}

It is worth mentioning why it should be hard or unlikely to show that ${\sf RwBQP}$ problems are solvable with rewindable IQP circuits.
In ${\sf RwBQP}$ computations, measurements and rewinding are allowed to be performed during the computations.
On the other hand, they are performed at the end of the computations in rewindable IQP circuits.
This implies that the principle of deferred measurement~\cite{NC00} is necessary to be held even when the rewinding operators are supplied.

Here, we would like to emphasize that our rewinding operator seems to not violate the no-signaling principle.
Although one may think that Theorem~\ref{rewindIQP} contradicts the no-signaling principal from the result in \cite{M14}, this is not the case.
\cite{M14} shows that if MBQC can be realized without byproduct operators on an entangled state shared by Alice and Bob, then Alice can send her message to Bob without any signal.
To apply their argument to our rewindable IQP circuits (i.e., MBQC with the rewinding), Alice has to rewind her measurements without Bob's qubits.
However, this is impossible because our rewinding operator defined in Def.~\ref{Rewinding} only works for pure states.

From~\cite{BJS11}, IQP circuits can solve any ${\sf PP}$ problem in quantum polynomial time under postselection.
Therefore, Theorem~\ref{rewindIQP} may imply an equivalence between the postselection and rewinding.
On the other hand, our argument on rewinding operators seems to be not applicable to DQC1~\cite{MFF14}, which can also solve any ${\sf PP}$ problem under postselection.
This is because the input state of DQC1 includes maximally mixed states.
This simple example may imply that the rewinding is weaker than the postselection in some situations.

\section{Conclusion}
\label{conclusion}
We have defined ${\sf RwBQP}$ and ${\sf CBQP}$ by using rewinding and cloning operators, respectively, and ${\sf AdPostBQP}$ by using postselections and have shown ${\sf BPP}^{\sf PP}\subseteq{\sf RwBQP}={\sf CBQP}={\sf AdPostBQP}\subseteq{\sf PSPACE}$.
To this end, we have proposed a mitigation protocol that exponentially reduces amplitudes of nontarget states.
To tighten the lower bound on ${\sf RwBQP}$, it may be useful to consider whether ${\sf PP}^{\sf PP}\subseteq{\sf RwBQP}^{\sf RwBQP}_{{\rm classical}}$.
This inclusion cannot be straightforwardly shown from Theorem~\ref{lemma2} because this theorem does not necessarily mean that any process in the ${\sf PP}$ machine can be efficiently simulated with the ${\sf RwBQP}$ machine.

Our mitigation protocol also implies that a polynomial number of non-adaptive postselections of outputs whose probabilities are polynomially close to one are sufficient for solving any ${\sf PostBQP}$ problem.
Then, we have considered how much rewinding operators are useful for restricted classes of quantum computation.
Particularly, we have shown that Clifford circuits are still classically simulatable even if a polynomial number of rewinding operators are supplied.
On the other hand, the rewinding promotes IQP circuits, which are believed to be neither universal nor classically simulatable, so that they can efficiently solve any ${\sf PP}$ problem.
To this end, we have shown that if the rewinding is possible, then universal quantum computation can be efficiently performed by measurement-based quantum computation (MBQC)~\cite{RB01} without adaptive measurements.
We have also shown that a single rewinding operator is sufficient to achieve tasks that are intractable for quantum computation under the post-quantum cryptographic assumption and ${\sf BQP}\nsupseteq{\sf SZK}$.

Some of our results imply an equivalence between the postselection and rewinding for uniform families of quantum circuits with pure input states.
How about non-uniform families and quantum interactive proof systems?
So far, several advice complexity classes with ${\sf PostBQP}$ and quantum {\red non-interactive} proof systems with a ${\sf PostBQP}$ verifier have been explored.
For example, the following results are already known:
\begin{enumerate}
\item ${\sf BQP/qpoly}\subseteq{\sf PostBQP/poly}\neq{\sf ALL}$~\cite{A04}
\item ${\sf PostBQP/qpoly}={\sf ALL}$~\cite{A06}
\item ${\sf PostQMA}={\sf PSPACE}$~\cite{MN17}
\item ${\sf PostQCMA}={\sf NP}^{\sf PP}$~\cite{MN17}
\item ${\sf PostQMA(2)}={\sf NEXP}$~\cite{B09,K18}
\end{enumerate}
Here, ${\sf PostQMA}$, ${\sf PostQCMA}$, and ${\sf PostQMA(2)}$ are defined by replacing a ${\sf BQP}$ verifier with a ${\sf PostBQP}$ one in ${\sf QMA}$, ${\sf QCMA}$, and ${\sf QMA(2)}$, respectively\footnote{In \cite{MN17}, ${\sf PostQMA}$ and ${\sf PostQCMA}$ are denoted by ${\sf QMA}_{\sf PostBQP}$ and ${\sf QCMA}_{\sf PostBQP}$, respectively. However, we use the terminology given in \cite{UHB17}}.
It would be interesting to consider what happens if we replace ${\sf PostBQP}$ with ${\sf RwBQP}$ or ${\sf AdPostBQP}$ in above complexity classes.
Since some classical description is necessary to do our rewinding, while it is not for the postselection, a difference between the postselection and rewinding may be observed in the settings of $2$, $3$, and $5$.
As another direction, it would be interesting to precisely characterize ${\sf RwBQP}(1)$ and investigate the computational capability of rewindable boson sampling~\cite{AA11,LLKROR14,HKSBSJ17} by defining a bosonic rewinding operator that inverts the Fock-basis measurement in a single mode.
It is also still open whether ${\sf RwBQP}$ is closed under composition even if quantum queries are allowed.

\section*{Acknowledgments}
We thank Yasuhiro Takahashi, Yusuke Aikawa, and Mikito Nanashima for helpful discussions.
We also thank Tomoyuki Morimae for fruitful discussions and pointing out \cite{ABFL14,ABFL16} to us.
YT is supported by the MEXT Quantum Leap Flagship Program (MEXT Q-LEAP) Grant Number JPMXS0118067394 and JPMXS0120319794 and JST [Moonshot R\&D -- MILLENNIA Program] Grant Number JPMJMS2061.
AM is supported by JST, ACT-X Grant Number JPMJAX210O, Japan.
ST is supported by the Grant-in-Aid for Transformative Research Areas No.JP20H05966 of JSPS, JST [Moonshot R\&D -- MILLENNIA Program] Grant Number JPMJMS2061, and the Grant-in-Aid for Scientific Research (A) No.JP22H00522 of JSPS.

\appendix
\section{Proofs of Corollaries~\ref{corollary}, \ref{corollary2}, and \ref{corollary3}}
\label{E}
The proof of Corollary~\ref{corollary} is as follows:
\begin{proof}
Since proofs are essentially the same for all three classes, we only write a concrete proof for ${\sf RwBQP}$.
Let $\overline{L}$ be the complement of $L$.
From Def.~\ref{BRQP}, when $x\in \overline{L}_{\rm yes}$ (i.e., $x\in L_{\rm no}$),
\begin{eqnarray}
\sum_{z\in A}\left|\left|\left(|1\rangle\langle 1|\otimes I^{\otimes n+m+\ell-1}\right)\left(X\otimes I^{\otimes n+m+\ell-1}\right)Q_n^{(z)}\left(|x\rangle|0^m\rangle|\tilde{\mathcal{D}}\rangle\right)\right|\right|^2\ge 2/3.
\end{eqnarray}
On the other hand, when $x\in \overline{L}_{\rm no}$ (i.e., $x\in L_{\rm yes}$),
\begin{eqnarray}
\sum_{z\in A}\left|\left|\left(|1\rangle\langle 1|\otimes I^{\otimes n+m+\ell-1}\right)\left(X\otimes I^{\otimes n+m+\ell-1}\right)Q_n^{(z)}\left(|x\rangle|0^m\rangle|\tilde{\mathcal{D}}\rangle\right)\right|\right|^2\le 1/3.
\end{eqnarray}
Therefore, ${\sf coRwBQP}\subseteq{\sf RwBQP}$.
By using the same argument, we can also show ${\sf coRwBQP}\supseteq{\sf RwBQP}$ and thus ${\sf RwBQP}={\sf coRwBQP}$.
\end{proof}

The proof of Corollary~\ref{corollary2} is as follows:
\begin{proof}
Since proofs are essentially the same for all three classes, we only write a concrete proof for ${\sf RwBQP}$.
By repeating the same ${\sf RwBQP}$ computation $m$ times and taking the majority vote on the outcomes, due to the Chernoff bound~\cite{AB09}, the error probability is improved from $1/3$ to $2^{-q(m)}$ for a positive polynomial function $q(m)$ in $m$.
Therefore, by setting $m$ so that $q(m)\ge p(n)$, we obtain this corollary.
\end{proof}

The proof of Corollary~\ref{corollary3} is as follows:
\begin{proof}
Since proofs are essentially the same for all three classes, we only write a concrete proof for ${\sf RwBQP}$.
From Def.~\ref{BRQP}, when a polynomial-time algorithm calls another polynomial-time algorithm as a subroutine, the resultant algorithm can still be realized in polynomial time.
Since the ${\sf RwBQP}$ computation has some error probability, a remaining concern is that errors may accumulate every time polynomial-time algorithms are called.
However, the accumulation of errors is negligible from Corollary~\ref{corollary2}.
As a result, we obtain ${\sf RwBQP}^{\sf RwBQP}_{{\rm classical}}={\sf RwBQP}$.
\end{proof}

\section{Proof of Lemma~\ref{lemmapspace}}
\label{C}
We give a proof of Lemma~\ref{lemmapspace}.
\begin{proof}
To show this lemma, it is sufficient to show that the acceptance probability
\begin{eqnarray}
p_{\rm acc}=\sum_{z\in A}q_z\left|\left|\left(|1\rangle\langle 1|\otimes I^{\otimes n+m+\ell-1}\right)\mathcal{N}[Q_n^{(z)}\left(|x\rangle|0^m\rangle|\tilde{\mathcal{D}}\rangle\right)]\right|\right|^2
\end{eqnarray}
can be computed in polynomial space.
To this end, we use the Feynman path integral.
Note that we only concretely show that
\begin{eqnarray}
\left|\left|\left(|1\rangle\langle 1|\otimes I^{\otimes n+m+\ell-1}\right)\mathcal{N}[Q_n^{(z)}\left(|x\rangle|0^m\rangle|\tilde{\mathcal{D}}\rangle\right)]\right|\right|^2
\end{eqnarray}
can be computed in polynomial space because the derivation of $q_z$ in polynomial space can be done in a similar way.
Let $k$ be some polynomial in $n$.
By using $q_i^{(z)}$ that is an elementary gate in a universal gate set or a single-qubit postselection onto $|1\rangle\langle 1|$ or $|0\rangle\langle 0|$ for $1\le i\le k$, we can decompose $\mathcal{N}[Q_n^{(z)}(\cdot)]$ as $\mathcal{N}[Q_n^{(z)}(\cdot)]=\prod_{i=1}^kq_i^{(z)}(\cdot)$ for any $z\in A$.
Let $N\equiv n+m+\ell$.
Therefore,
\begin{eqnarray}
\nonumber
&&\left(|1\rangle\langle 1|\otimes I^{\otimes N-1}\right)\mathcal{N}[Q_n^{(z)}\left(|x\rangle|0^m\rangle|\tilde{\mathcal{D}}\rangle\right)]\\
\nonumber
&=&\left[|1\rangle\langle 1|\otimes \left(\sum_{d\in\{0,1\}^{N-1}}|d\rangle\langle d|\right)\right]q_k^{(z)}\left[\prod_{i=1}^{k-1}\left(\sum_{s^{(i)}\in\{0,1\}^N}|s^{(i)}\rangle\langle s^{(i)}|\right)q_i^{(z)}\right]|x\rangle|0^m\rangle|\tilde{\mathcal{D}}\rangle.\\
\label{pathpacc}
\end{eqnarray}
Let $s$ be a shorthand notation of a $(k-1)N$-bit string $s^{(1)}s^{(2)}\ldots s^{(k-1)}$.
By defining
\begin{eqnarray}
g(s,d,z)\equiv\langle 1|\langle d|\left[q_k^{(z)}\left(\prod_{i=1}^{k-1}|s^{(i)}\rangle\langle s^{(i)}|q_i^{(z)}\right)\right]|x\rangle|0^m\rangle|\tilde{\mathcal{D}}\rangle,
\end{eqnarray}
$p_{\rm acc}$ can be written as
\begin{eqnarray}
\label{finalpath}
\sum_{z\in A}q_z\sum_{s,\tilde{s}\in\{0,1\}^{(k-1)N},d\in\{0,1\}^{N-1}}g(s,d,z)g^{\ast}(\tilde{s},d,z).
\end{eqnarray}
{\red Since} $q_i^{(z)}$ is just a constant-size matrix, each term $g(s,d,z)g^{\ast}(\tilde{s},d,z)$ can be computed in polynomial space\footnote{It may not be able to be computed in polynomial time, because $q_i^{(z)}$ may be a postselection.}.
Therefore, Eq.~(\ref{finalpath}) can also be computed in polynomial space.
\end{proof}

\section{Success Probability of Our Mitigation Protocol}
\label{D}
We show that the probability that the state in Eq.~(\ref{ideal}) is output in our mitigation protocol is at least $1-5n/8^n$ and that $(1-p_{2n+3})/p_{2n+3}\ge 1$.
For clarity, we show a schematic diagram of our mitigation protocol in Fig.~\ref{miti}.

When the outcome $0$ is obtained by measuring the last register of Eq.~(\ref{prepare}) in the $Z$ basis, the amplitude of the nontarget state $|\psi_t^\perp\rangle$ is mitigated with the factor $1/\sqrt{2}$ (up to normalization) because
\begin{eqnarray}
\label{projectionTOCS}
\left(I^{\otimes 2}\otimes \langle 0|\right)\left(\sqrt{p_i}|\psi_t^\perp\rangle|+\rangle+\sqrt{1-p_i}|\psi_t\rangle|0\rangle\right)=\sqrt{\frac{p_i}{2}}|\psi_t^\perp\rangle+\sqrt{1-p_i}|\psi_t\rangle.
\end{eqnarray}
Therefore, for any $i$, the probability $q_i$ that the outcome $0$ is obtained by measuring the last register of Eq.~(\ref{prepare}) in the $Z$ basis is
\begin{eqnarray}
q_i=\frac{p2^{-(i+1)}+(1-p)}{1-(1-2^{-i})p},
\end{eqnarray}
where we have used $p_0=p$.
Therefore, for any $i$, the probability that we obtain $0$ by measuring the last register of Eq.~(\ref{prepare}) in the $Z$ basis before or at $c=3n$ is
\begin{eqnarray}
\label{suci}
1-\left(1-q_i\right)^{3n}\ge1-\left(1-q_0\right)^{3n}=1-\left(\frac{p}{2}\right)^{3n}.
\end{eqnarray}

Our purpose is to sufficiently mitigate the amplitude of $|\psi_t^\perp\rangle$ so that we obtain the outcome $1$ by measuring the second register of Eq.~(\ref{i+1}) in the $Z$ basis with probability of at least $1/2$.
To this end, it is sufficient to run our mitigation protocol until $i=N$ such that
\begin{eqnarray}
\frac{1-p_N}{p_N}\ge 1.
\end{eqnarray}
From Eq.~(\ref{suc2}), this condition can be satisfied by setting
\begin{eqnarray}
\label{N}
N=\left\lceil\log_2{\left(\frac{p}{1-p}\right)}\right\rceil.
\end{eqnarray}

By combining Eqs.~(\ref{suci}) and (\ref{N}), the probability that the state in Eq.~(\ref{ideal}) is output in our mitigation protocol (i.e., the probability of our mitigation protocol reaching to {\bf Success} in Fig.~\ref{miti}) is at least
\begin{eqnarray}
\left[1-\left(\frac{p}{2}\right)^{3n}\right]^N&\ge& 1-\left[\log_2{\left(\frac{p}{1-p}\right)}+1\right]\left(\frac{p}{2}\right)^{3n}\\
\label{lower1}
&\ge&1-\left[\log_2{\left(\frac{p}{1-p}\right)}+1\right]\left(\frac{1}{2}\right)^{3n}.
\end{eqnarray}
Since $\log_2{[p/(1-p)]}$ is a monotonically increasing function of $p$ in the range of $0<p<1$, the remaining task is to upper bound $p$.
(Recall that our goal in this appendix is to show that Eq.(\ref{lower1}) is lower bounded by $1-5n/8^n$ and $N\le 2n+3$.)

From the simple observation that $2[(2^n-s)^2+s^2]$, which is the denominator of $p$ in Eq.~(\ref{p}), is a symmetric convex downward function that becomes minimum at $s=2^{n-1}$, and that $(2^n-s)^2$ in the numerator of $p$ is a monotonically decreasing function in the range of $0<s\le 2^n$, the value of $s$ maximizing $p$ (for any $\alpha$, $\beta$, and $n$) is between $1$ and $2^{n-1}$.
To upper bound $p$, we separately consider three cases: (i) $1\le s\le (1-1/\sqrt{2})2^n$, (ii) $(1-1/\sqrt{2})2^n<s\le 2^{n-1}-1$, and (iii) $s=2^{n-1}$.

(i) When $1\le s\le (1-1/\sqrt{2})2^n$, the inequality $2(2^n-s)^2\ge 4^n$ holds.
From Eq.~(\ref{p}),
\begin{eqnarray*}
p=\frac{2\alpha^2(2^n-s)^2+\beta^24^n}{2[(2^n-s)^2+s^2]}.
\end{eqnarray*}
Since $\alpha^2+\beta^2=1$, the numerator of $p$ is upper bounded by $2\alpha^2(2^n-s)^2+2\beta^2(2^n-s)^2=2(2^n-s)^2$.
Therefore,
\begin{eqnarray}
\label{(i)}
p\le\frac{(2^n-s)^2}{(2^n-s)^2+s^2}\le\frac{(2^n-1)^2}{(2^n-1)^2+1}=1-\frac{1}{(2^n-1)^2+1}.
\end{eqnarray}

(ii) When $(1-1/\sqrt{2})2^n<s\le 2^{n-1}-1$, the inequality $2(2^n-s)^2<4^n$ holds, and hence
\begin{eqnarray}
\label{(ii)}
p\le\frac{4^n}{2[(2^n-s)^2+s^2]}\le\frac{4^n}{2[(2^{n-1}+1)^2+(2^{n-1}-1)^2]}=1-\frac{1}{4^{n-1}+1}.
\end{eqnarray}
The first inequality can be derived in a similar way as in (i).

(iii) When $s=2^{n-1}$,
\begin{eqnarray}
\label{(iii)}
p=\frac{\alpha^2}{2}+\beta^2\le1-\frac{1}{2(4^n+1)},
\end{eqnarray}
where we have used $2^{-n}\le\beta/\alpha\le2^n$ and $\alpha^2+\beta^2=1$.

From Eqs.~(\ref{(i)}), (\ref{(ii)}), and (\ref{(iii)}), $p\le 1-1/[2(4^n+1)]\equiv p_{\rm max}$, and hence
\begin{eqnarray}
\label{ubTOCS}
N\le\log_2{\frac{p}{1-p}}+1\le\log_2{\frac{p_{\rm max}}{1-p_{\rm max}}}+1\le 2n+3.
\end{eqnarray}
This implies that Eq.(\ref{lower1}) is lower bounded by
\begin{eqnarray}
1-\left[\log_2{\left(\frac{p_{\rm max}}{1-p_{\rm max}}\right)}+1\right]\left(\frac{1}{2}\right)^{3n}\ge 1-\frac{2n+3}{2^{3n}}\ge 1-\frac{5n}{8^n}.
\end{eqnarray}

\section{Proof of Theorem~\ref{easy}}
\label{G3}
The proof of Theorem~\ref{easy} is as follows:
\begin{proof}
To solve the problem in Def.~\ref{problem} with a constant probability, we use the idea used in \cite{CCKW19}.
We prepare the state
\begin{eqnarray}
\label{firststate}
\frac{\sum_{s\in\mathbb{Z}_q^n,e\in\chi^m,d\in\{0,1\}}|s,e,d\rangle|f_K(s,e,d)\rangle}{\sqrt{2q^n(2\mu+1)^m}},
\end{eqnarray}
where $1/\sqrt{2q^n(2\mu+1)^m}$ is the normalization factor (see Theorem~\ref{crf}).
When there exists a natural number $N$ satisfying $2^N=2q^n(2\mu+1)^m$, this preparation is trivially possible in quantum polynomial time with unit probability.
If this is not the case, we prepare
\begin{eqnarray}
\frac{\sum_{(s,e,d)\in\mathbb{Z}_q^n\times\chi^m\times\{0,1\}}|s,e,d\rangle|f_K(s,e,d)\rangle|1\rangle+\sum_{(s,e,d)\notin\mathbb{Z}_q^n\times\chi^m\times\{0,1\}}|s,e,d\rangle|0^{m\log_2{q}}\rangle|0\rangle}{\sqrt{2^{\tilde{N}}}}
\end{eqnarray}
with unit probability, where $\tilde{N}$ is the smallest natural number satisfying $2^{\tilde{N}}\ge2q^n(2\mu+1)^m$.
If we obtain the outcome $1$ by measuring the third register in the computational basis, we can prepare the state in Eq.~(\ref{firststate}).
From $2q^n(2\mu+1)^m>2^{\tilde{N}-1}$, the probability of $1$ being observed is larger than $1/2$.
Therefore, by repeating these procedures, we can obtain the outcome $1$ at least once with probability of at least $1-o(1)$.

By measuring the second register in Eq.~(\ref{firststate}), we obtain a value of $f_K(s,e,d)$.
From the $\delta$-$2$ regularity of $\mathcal{F}$, the obtained output $f_K(s,e,d)$ has exactly two different preimages with probability of at least $\delta$.
When $f_K(s,e,d)$ has exactly two different preimages, the state of the first register becomes
\begin{eqnarray}
\label{two}
\frac{|s,e,1\rangle+|s+s_0,e+e_0,0\rangle}{\sqrt{2}},
\end{eqnarray}
where $f_K(s,e,1)=f_K(s+s_0,e+e_0,0)$.
Then, we measure the state in Eq.~(\ref{two}) and obtain the values of $(s,e,1)$ or $(s+s_0,e+e_0,0)$.

To obtain the other one with probability $1/2$, we would like to obtain the state in Eq.~(\ref{two}) again.
It is possible by applying the rewinding operator $R$ on $|s,e,1\rangle$ or $|s+s_0,e+e_0,0\rangle$ and a classical description\footnote{More precisely, the classical description means a transcript of how to prepare the state in Eq.~(\ref{two}) from a tensor product of $|0\rangle$'s. Let $V$ be an operator that prepares the state in Eq.~(\ref{firststate}) from $|0\rangle$'s and $\ell$ be the number of qubits required in the first register in Eq.~(\ref{firststate}). Then, the classical description is $(I^{\otimes\ell}\otimes|f_K(s,e,d)\rangle\langle f_K(s,e,d)|)V$. Note that $V$ can be decomposed into a polynomial number of elementary gates in a universal gate set and the postselection onto $|1\rangle\langle 1|$.} of the state in Eq.~(\ref{two}).
As an important point, since the state in Eq.~(\ref{two}) becomes $|s,e,1\rangle$ or $|s+s_0,e+e_0,0\rangle$ by measuring only the last single qubit in the $Z$ basis, a single rewinding operator is sufficient to rewind it.

On the other hand, if rewinding operator is not allowed, the probability of the problem being solved is super polynomially small from the collision resistance of the function family $\mathcal{F}$.
\end{proof}

\section{Proof of Theorem~\ref{newlemma}}
\label{newH}
In this appendix, we show that the ${\sf RwBQP}(1)$ machine can solve the SD problem with probabilities at least $1/2-2^{-O(n^c)}$ and $1-2\cdot2^{-O(n^c)}$ when $D_{\rm TV}(P_0,P_1)<2^{-O(n^c)}$ and $D_{\rm TV}(P_0,P_1)>1-2^{-O(n^c)}$, respectively.
To this end, we use an argument inspired by \cite{ABFL16}\footnote{As a difference between their argument in \cite{ABFL16} and ours, we replace their non-collapsing measurement with a single rewinding operator and an ordinary (i.e., a collapsing) measurement. Furthermore, although they use three non-collapsing measurements, we can perform the rewinding operator only once.}.
First, the ${\sf RwBQP}(1)$ machine prepares
\begin{eqnarray}
\frac{1}{\sqrt{2^{n+1}}}\sum_{b\in\{0,1\},x\in\{0,1\}^n}|b\rangle|x\rangle|C_b(x)\rangle.
\end{eqnarray}
By measuring the last register in the computational basis, it obtains the outcome $y\in\{0,1\}^m$ and
\begin{eqnarray}
\label{newy}
\frac{|0\rangle\left(\sum_{x: C_0(x)=y}|x\rangle\right)+|1\rangle\left(\sum_{x: C_1(x)=y}|x\rangle\right)}{\sqrt{2^n(P_0(y)+P_1(y))}},
\end{eqnarray}
where for $b\in\{0,1\}$, $P_b(y)=|\{x\in\{0,1\}^n: C_b(x)=y\}|/2^n$ is the probability of $C_b$ outputting $y$ for uniformly random inputs $x\in\{0,1\}^n$.
This event occurs with probability $(P_0(y)+P_1(y))/2$.
Then, it measures the first register in Eq.~(\ref{newy}) in the computational basis and obtain the outcome $b_1\in\{0,1\}$.
By using a single rewinding operator, it can perform the same measurement again and obtain another outcome $b_2\in\{0,1\}$.
Finally, it outputs $1$ if $b_1\neq b_2$.
Otherwise, it outputs $0$.

We now calculate error probabilities, i.e., probabilities of the machine outputting $0$ and $1$ when $D_{\rm TV}(P_0,P_1)<2^{-O(n^c)}$ and $D_{\rm TV}(P_0,P_1)>1-2^{-O(n^c)}$, respectively.
First, we consider the case of $D_{\rm TV}(P_0,P_1)<2^{-O(n^c)}$.
The probability $p_{\rm err}$ of the machine outputting $0$, i.e., that of $b_1=b_2$ is
\begin{eqnarray}
\label{newerr}
p_{\rm err}=\sum_{y\in\{0,1\}^m}\frac{P_0(y)+P_1(y)}{2}\frac{P_0(y)^2+P_1(y)^2}{(P_0(y)+P_1(y))^2}{\red=\sum_{y\in\{0,1\}^m}\frac{P_0(y)^2+P_1(y)^2}{2(P_0(y)+P_1(y))}}.
\end{eqnarray}
Let $\delta(y)\equiv{\rm max}\{P_0(y)-P_1(y),P_1(y)-P_0(y)\}$ and $P_{\rm min}(y)\equiv{\rm min}\{P_0(y),P_1(y)\}$.
From Eq.~(\ref{newerr}),
\begin{eqnarray}
\label{newc}
p_{\rm err}=\frac{1}{2}\left(1+\sum_{y\in\{0,1\}^m}\delta(y)\frac{P_{\rm min}(y)+\delta(y)}{2P_{\rm min}(y)+\delta(y)}\right)\le\frac{1}{2}\left(1+\sum_{y\in\{0,1\}^m}\delta(y)\right)<\frac{1}{2}+2^{-O(n^c)},
\end{eqnarray}
where we have used $\sum_{y\in\{0,1\}^m}\delta(y)=2D_{\rm TV}(P_0,P_1)$ in the last inequality.

Then, we consider the case of $D_{\rm TV}(P_0,P_1)>1-2^{-O(n^c)}$.
The probability $p'_{\rm err}$ of the machine outputting $1$, i.e., that of $b_1\neq b_2$ is
\begin{eqnarray}
\label{newerr2}
p'_{\rm err}=\sum_{y\in\{0,1\}^m}\frac{P_0(y)+P_1(y)}{2}\frac{2P_0(y)P_1(y)}{(P_0(y)+P_1(y))^2}{\red=\sum_{y\in\{0,1\}^m}\frac{P_0(y)P_1(y)}{P_0(y)+P_1(y)}}.
\end{eqnarray}
Since $D_{\rm TV}(P_0,P_1)>1-2^{-O(n^c)}$, there exists a set $S$ such that $\sum_{y\in S}P_0(y)\ge 1-2^{-O(n^c)}$ and $\sum_{y\in S}P_1(y)\le 2^{-O(n^c)}$.
Let $\bar{S}$ be a complement of $S$.
From Eq.~(\ref{newerr2}),
\begin{eqnarray}
p'_{\rm err}=\sum_{y\in S}\frac{P_0(y)P_1(y)}{P_0(y)+P_1(y)}+\sum_{y\in \bar{S}}\frac{P_0(y)P_1(y)}{P_0(y)+P_1(y)}\le\sum_{y\in S}P_1(y)+\sum_{y\in\bar{S}}P_0(y)\le 2\cdot2^{-O(n^c)},
\end{eqnarray}
where we have used $\sum_{y\in S\cup\bar{S}}P_0(y)=1$ in the last inequality.

\end{document}